\documentclass[10pt]{article}

\usepackage{url}
\usepackage{mathtools}
\usepackage{amssymb}
\usepackage{amsthm}
\usepackage{latexsym}
\usepackage[shortlabels]{enumitem}
\usepackage{dsfont}
\usepackage{appendix}
\usepackage{color}
\usepackage[utf8]{inputenc}
\usepackage[T1]{fontenc}
\usepackage{geometry}
\usepackage{todonotes}
\usepackage{lmodern}
\usepackage{anyfontsize}
\usepackage{stmaryrd}
\usepackage{comment}
\usepackage{bm}
\usepackage{mathrsfs}
\usepackage{mdframed}
\usepackage{natbib}
\usepackage[english]{babel}
\usepackage{cases}
\usepackage{graphicx,subcaption}
\graphicspath{{fig/}}
\usepackage{booktabs,tabularx}
\usepackage{aliascnt}
\newcolumntype{Y}{>{\raggedright\arraybackslash}X}

\usepackage{booktabs,siunitx}
\mdfsetup{middlelinecolor=blue,middlelinewidth=2pt,linewidth=0pt,backgroundcolor=blue!15,,roundcorner=10pt}

\allowdisplaybreaks[1]
\setcitestyle{numbers,open={[},close={]}}

\definecolor{red}{rgb}{0.7,0.15,0.15}
\definecolor{green}{rgb}{0,0.5,0}
\definecolor{blue}{rgb}{0,0,0.7}

\usepackage{hyperref}
\hypersetup{colorlinks, linkcolor={red}, citecolor={green}, urlcolor={blue}}
\numberwithin{equation}{section}
\usepackage[english]{cleveref}

\makeatletter \@addtoreset{equation}{section}

\newtheorem{theorem}{Theorem}[section]

\newaliascnt{assumption}{theorem}

\aliascntresetthe{assumption}

\newaliascnt{proposition}{theorem}
\newtheorem{proposition}[proposition]{Proposition}
\aliascntresetthe{proposition}

\newaliascnt{definition}{theorem}

\aliascntresetthe{definition}

\newaliascnt{lemma}{theorem}
\newtheorem{lemma}[lemma]{Lemma}
\aliascntresetthe{lemma}

\newaliascnt{example}{theorem}

\aliascntresetthe{example}

\newaliascnt{corollary}{theorem}

\aliascntresetthe{corollary}

\newaliascnt{remark}{theorem}
\newtheorem{remark}[remark]{Remark}
\aliascntresetthe{remark}

\newaliascnt{condition}{theorem}

\aliascntresetthe{condition}

\crefname{theorem}{theorem}{theorems}
\Crefname{theorem}{Theorem}{Theorems}

\crefname{assumption}{assumption}{assumptions}
\Crefname{assumption}{Assumption}{Assumptions}

\crefname{proposition}{proposition}{propositions}
\Crefname{proposition}{Proposition}{Propositions}

\crefname{definition}{definition}{definitions}
\Crefname{definition}{Definition}{Definitions}

\crefname{lemma}{lemma}{lemmas}
\Crefname{lemma}{Lemma}{Lemmas}

\crefname{example}{example}{examples}
\Crefname{example}{Example}{Examples}

\crefname{corollary}{corollary}{corollaries}
\Crefname{corollary}{corollary}{Corollaries}

\crefname{remark}{remark}{remarks}
\Crefname{remark}{remark}{Remarks}

\crefname{condition}{condition}{conditions}
\Crefname{condition}{Condition}{Conditions}

\newcommand\frF{F}
\newcommand\frG{G}

\newcommand{\smallertext}[1]{\text{\fontsize{5}{5}\selectfont$#1$}}
\newcommand{\smalltext}[1]{\text{\fontsize{4}{4}\selectfont$#1$}}
\newcommand{\tinytext}[1]{\text{\fontsize{3}{3}\selectfont$#1$}}

\newcommand{\vertiii}[1]{{\left\vert\kern-0.25ex\left\vert\kern-0.25ex\left\vert#1 \right\vert\kern-0.25ex\right\vert\kern-0.25ex\right\vert}}

\def\E{\mathbb{E}}
\def\F{\mathbb{F}}

\def\N{\mathbb{N}}
\def\P{\mathbb{P}}

\def\R{\mathbb{R}}

\def\Ac{\mathcal{A}}

\def\Ec{\mathcal{E}}
\def\Fc{\mathcal{F}}

\def\Oc{\mathcal{O}}

\def\Sc{\mathcal{S}}

\def\Zc{\mathcal{Z}}

\def\d{\mathrm{d}}

\def\Sum{\displaystyle\sum}%%%%%%%%%
 \newcommand{\Diag}{\mathrm{Diag}}
 
 \newcommand{\tr}{\mathrm{Tr}}

\DeclareMathOperator{\spr}{spr}

\makeatletter
\newcommand*\bigcdot{\mathpalette\bigcdot@{.85}}
\newcommand*\bigcdot@[2]{\mathbin{\vcenter{\hbox{\scalebox{#2}{$\m@th#1\bullet$}}}}}

\begin{document}

	\title{Optimal incentive scheme for ESG disclosure}

	\author{Imen {\sc{Ben Tahar}}\footnote{CEREMADE, Université Paris Dauphine--PSL and Finance for Energy Market Research Centre
(FIME), Paris, France, imen@ceremade.dauphine.fr.}
			\and Dylan {\sc{Possama{\"{i}}}}\footnote{ETH~Z{\"{u}}rich, Department of Mathematics, Switzerland, dylan.possamai@math.ethz.ch. This author gratefully acknowledges partial support by the SNF project MINT 205121-21981.}
			\and Xiaolu {\sc{Tan}}\footnote{Department of Mathematics, The Chinese University of Hong Kong. xiaolu.tan@cuhk.edu.hk. This author gratefully acknowledges support by Hong Kong RGC General Research Fund (projects 14302622) and the faculty direct grant.}}

	\date{\today}

	\maketitle

	\begin{abstract}
This paper characterises optimal incentive schemes for ESG disclosure in a continuous-time principal--agent setting. We model a risk-averse principal (\emph{e.g.}, a platform or standard-setter) contracting with a team of heterogeneous agents whose disclosure signals are each correlated with a traded climate risk factor. The optimal contract balances incentive provision against the variance of aggregate payouts by leveraging three instruments: own-signal loading, cross-signal loadings across agents, and hedging tilts on the traded asset. We derive closed-form linear optimal controls in a tractable linear--quadratic--Gaussian framework. When the principal is nearly risk-neutral, the contract uses the traded asset purely to hedge the specific `enforcement risk' generated by high-powered incentives. As the principal’s risk aversion increases, the optimal scheme converges to a `market-neutral' regime where aggregate asset exposure is eliminated and the cross-signal structure tightens to an `identity pooling' constraint. We characterise this limit analytically as a constrained quadratic program governed by an M-matrix. In the high-risk-aversion regime, heterogeneity creates genuinely new effects absent under symmetry: the cross-section of $S$-tilts must change sign (unless degenerate), and an agent's own-signal diagonal $z^{\smallertext{Q},i,i,\star}$ can turn negative when that row is too strongly exposed to the common traded factor relative to the rest of the group. The results provide a theoretical foundation for `mixed' compensation structures in Regenerative Finance (ReFi), rationalising the use of both stable payments and volatile governance tokens to optimise risk-sharing.

		\vspace{5mm}
	\end{abstract}

\section{Introduction}
\label{sec:intro}

Sustainability disclosure has become important for both financial markets and corporate decisions. Institutional investors use climate and ESG information in portfolio allocation and stewardship, see \citeauthor*{krueger2020importance} \cite{krueger2020importance}. Carbon exposures are priced in the cross-section of returns, see \citeauthor*{bolton2021do} \cite{bolton2021do}, and investors demand compensation for carbon tail risk, see \citeauthor*{ilhan2021carbon} \cite{ilhan2021carbon}. In equilibrium, sustainability preferences affect asset prices, expected returns, and firms' costs of capital, as in \citeauthor*{pastor2021sustainable} \cite{pastor2021sustainable}. Firms respond through investment choices and financing instruments, including green bonds, see \citeauthor*{flammer2021corporate} \cite{flammer2021corporate} and \citeauthor*{baldacci2022governmental} \cite{baldacci2022governmental}, and through greener investment policies under investor screening, see \citeauthor*{heinkel2001effect} \cite{heinkel2001effect}. These developments take place alongside new sustainability reporting rules, including the ISSB standards IFRS~S1 and IFRS~S2 and the European Union's CSRD/ESRS framework \cite{ifrsS1_2023,ifrsS2_2023,ec_csrd_2023,eu_esrs_2023}. They also connect to the broader debate on the objective of the firm, see \citeauthor*{hart2017companies} \cite{hart2017companies}.

\medskip
The difficulty is that ESG information is not a single clean variable. Environmental, social, and governance measures differ in observability, auditability, and timing. Some are close to physical quantities, such as emissions or energy use; others depend on internal processes, supply chains, or long-run transition plans. These signals are noisy, and they may be interpreted differently by investors, firms, rating agencies, and regulators. This helps explain why ESG ratings often disagree, as documented by \citeauthor*{berg2022aggregate} \cite{berg2022aggregate}, and why the same disclosure can be read differently depending on prior beliefs and salience, as in \citeauthor*{christensen2022why} \cite{christensen2022why}. These measurement frictions make incentive design difficult: stronger incentives can improve disclosure, but they also expose agents and principals to more risk.

\medskip
Our model focuses on one specific source of risk. In practice, ESG signals may be directly correlated with each other. Here, however, we deliberately abstract from direct cross-correlation among the signal processes $(Q^i)_{i\in\{1,\dots,n\}}$. Instead, each signal is correlated with a common traded factor $S$. This traded factor can be interpreted as a marketable climate-risk factor, a carbon-credit price, or a token whose value is exposed to climate-related risk. Thus agents are linked through their heterogeneous correlations with $S$, through the principal's objective over average compensation, and through the fact that the contract may load on several agents' signals at once.

\medskip
The paper studies how a principal should design disclosure incentives in this environment. The principal contracts with $n$ agents. Agent $i$ controls the drift of a signal $Q^i$, which represents the quality or quantity of ESG information produced by that agent. The principal can write contracts that depend on the whole signal vector $Q$ and on the traded factor $S$. The contract of agent $i$ may therefore include three types of exposure: a loading on his own signal $Q^i$, off-diagonal loadings on other signals $Q^j$, and a loading on the traded factor $S$. Agents are risk-averse and heterogeneous in effort costs, risk aversion, signal scale, and correlation with $S$. The principal is also risk-averse and evaluates average payoff across the team.

\medskip
The main trade-off is simple. Loading on $Q^i$ gives incentives to agent $i$, but it also makes compensation risky. Loading on $S$ can hedge part of this risk because each $Q^i$ is correlated with $S$. However, aggregate exposure to $S$ creates risk for the principal. The off-diagonal signal loadings help redistribute risk across contracts while preserving incentives. The problem is therefore to choose the own-signal loadings, the off-diagonal signal loadings, and the $S$-tilts jointly.

\medskip
We solve this problem in a linear--quadratic--Gaussian setting. The optimal contract is linear in the terminal signal increments and in the log-return of $S$. The optimal coefficients are obtained from an explicit quadratic maximisation problem. In the homogeneous case, we derive closed-form formulas for all coefficients and obtain their signs and limits. In the heterogeneous case, we give two transparent limiting regimes: the nearly risk-neutral principal and the highly risk-averse principal.

\medskip
The first regime is the limit $\gamma_{\smallertext{\rm P}}\longrightarrow0$. In this case, the principal does not penalise aggregate risk strongly, and the contracts decouple across rows. The $S$-tilt of agent $i$ hedges the risk created by his incentive exposure, so it has the opposite sign to $\rho_i$. The diagonal loading on $Q^i$ is positive. The off-diagonal loading of contract $i$ on signal $Q^j$ has the sign of $\rho_i\rho_j$. Thus, in the nearly risk-neutral case, the contract uses $S$ mainly as a hedge for the risk created by strong incentives.

\medskip
The second regime is the limit $\gamma_{\smallertext{\rm P}}\longrightarrow\infty$. In this case, the principal strongly penalises aggregate exposure. The solution converges to a constrained problem in which
\[
\sum_{i=1}^n z^{\smallertext S,i}=0,
\;
\sum_{i=1}^n z^{\smallertext Q,i,j}=1,\; j\in\{1,\dots,n\}.
\]
The first condition is market neutrality: aggregate exposure to the traded factor is eliminated. The second condition is an identity-pooling restriction: the total loading on each signal is fixed at one. We solve this constrained problem explicitly. After eliminating the signal-load multipliers, the limiting vector of $S$-tilts is characterised by an $n\times n$ linear system whose matrix is a nonsingular $M$-matrix. This gives both existence and sign information.

\medskip
A key implication is that heterogeneity matters. Under symmetry, the limiting aggregate $S$-tilt is zero because every individual $S$-tilt is zero. Under heterogeneity, this is no longer true. If the limiting vector of $S$-tilts is not identically zero, it must contain both positive and negative entries. Thus some agents are used to offset the common-factor exposure of others. This rebalancing can also change the sign of a diagonal signal loading. In the unconstrained signed-action model, an agent whose exposure to the traded factor is large relative to the rest of the group may receive a negative own-signal loading in the high-risk-aversion regime. This should not be read as a standard bonus reduction. Since the induced action is proportional to the diagonal loading, a negative diagonal means that the model prescribes a sign reversal in the action-loading. Economically, this is a malus-type effect generated by aggregate risk control.

\medskip
The paper relates to several literatures. On disclosure, it builds on the idea that information release depends on incentives, risk, and proprietary costs, as in \citeauthor*{dye1985disclosure} \cite{dye1985disclosure} and \citeauthor*{verrecchia2001essays} \cite{verrecchia2001essays}, and on the link between information quality and the cost of capital, see \citeauthor*{healy2001information} \cite{healy2001information}, \citeauthor*{lambert2007accounting} \cite{lambert2007accounting}, and \citeauthor*{leuzi2016economics} \cite{leuzi2016economics}. On incentives, it connects to moral hazard, teams, and relative performance evaluation, starting with \citeauthor*{holmstrom1979moral} \cite{holmstrom1979moral,holmstrom1982moral}, \citeauthor*{mookherjee1984optimal} \cite{mookherjee1984optimal}, and \citeauthor*{gibbons1990relative} \cite{gibbons1990relative}. It also relates to the linear-contract benchmark of \citeauthor*{holmstrom1987aggregation} \cite{holmstrom1987aggregation}. In continuous time, the paper is close to \citeauthor*{sannikov2008continuous} \cite{sannikov2008continuous}, \citeauthor*{williams2008dynamic} \cite{williams2008dynamic}, and to Brownian principal--agent models with several agents, including \citeauthor*{koo2008optimal} \cite{koo2008optimal}, \citeauthor*{cvitanic2017moral} \cite{cvitanic2017moral,cvitanic2018dynamic}, \citeauthor*{elie2019contracting} \cite{elie2019contracting}, \citeauthor*{elie2019tale} \cite{elie2019tale}, \citeauthor*{elie2021mean} \cite{elie2021mean}, \citeauthor*{baldacci2021optimal} \cite{baldacci2021optimal}, and \citeauthor*{hernandez2024principal} \cite{hernandez2024principal}.

\medskip
One distinction is worth stressing. The off-diagonal signal loadings in our model are not classical relative-performance evaluation terms. Since the processes $Q^1,\dots,Q^n$ are independent, these loadings do not remove a common output-noise component from the signals themselves. Instead, they are used to manage the common traded-factor risk created by the correlations between each $Q^i$ and $S$. In this sense, they are cross-signal risk-sharing terms rather than standard peer-benchmarking terms.

\medskip
The model also speaks to blockchain-based disclosure and token-based reward systems, such as \citeauthor*{niu2024blockchain} \cite{niu2024blockchain}. In those settings, tokens and on-chain records can help make disclosure rewards verifiable and enforceable. Our model adds a risk-sharing layer: if token values are correlated with climate or carbon-market risk, then paying only in the token may impose too much systematic risk on agents. The optimal contract instead combines signal-based compensation with a carefully chosen exposure to the traded factor. This interpretation is relevant for Regenerative Finance (ReFi), where tokenomics is often used to connect environmental data with market value, see \citeauthor*{sorensen2023tokenized} \cite{sorensen2023tokenized}. It also relates to work on blockchain-based ESG disclosure and tokenised carbon credits, see \citeauthor*{harish2023blockchain} \cite{harish2023blockchain} and \citeauthor*{ballesteros2024tokenized} \cite{ballesteros2024tokenized}. In the carbon-sequestration example, $Q$ can represent oracle-verified project data, whose verification is technically difficult and delayed, see \citeauthor*{caldarelli2020understanding} \cite{caldarelli2020understanding}, while $S$ can represent a liquid carbon-credit or governance-token price.

\medskip
The contract should be interpreted as a net transfer claim indexed to $Q$ and $S$. The model allows signed exposures and does not impose limited liability or collateral constraints. Thus the results should not be read literally as long-only token grants. Rather, they describe the optimal risk-sharing benchmark in a setting where the principal can use both signal-based payments and traded-factor exposures.

\medskip
The main contributions are as follows. First, we derive closed-form optimal disclosure sensitivities in homogeneous teams, including signs, limits, and comparative statics. Second, we solve the heterogeneous high-risk-aversion limit as an explicit constrained quadratic program and show that the reduced system for the limiting $S$-tilts is governed by a nonsingular $M$-matrix. Third, we show that heterogeneity can generate mixed-sign $S$-tilts and, in the signed-action model, negative own-signal loadings in the high-risk-aversion regime.

\medskip
The rest of the paper is organised as follows. \Cref{sec:model} introduces the contracting model and reduces the principal's problem to a quadratic optimisation problem. \Cref{sec:homogeneous} studies the maximiser of this quadratic problem. It first treats the homogeneous case, then analyses the heterogeneous limits as $\gamma_\smallertext{\rm P}\longrightarrow 0$ and $\gamma_\smallertext{\rm P}\longrightarrow\infty$, and finally discusses the economic interpretation and numerical illustrations.
\medskip

{\small\paragraph{Notation and conventions.}
Throughout the paper, all vectors are understood as \emph{column vectors} unless explicitly stated otherwise. We write
\[
\N\coloneqq \{0,1,2,\dots\},
\;
\N^\star\coloneqq \N\setminus\{0\},
\;
\R^n\equiv \R^{n\times 1},
\]
and, for $(n,p)\in(\N^\star)^2$, $\R^{n\times p}$ for the space of real matrices with $n$ rows and $p$ columns. For $x\in\R^n$, we write $x_i$ for its $i$th component, $i\in\{1,\dots,n\}$. and for $A\in\R^{n\times p}$ we write $A^{i,j}$ or $A_{i,j}$ for its $(i,j)$ entry, $A^{i,:}$ for its $i$th row, and $A^{:,j}$ for its $j$th column, $(i,j)\in\{1,\dots,n\}\times\{1,\dots,p\}$. The transpose of a matrix $A$ is denoted by $A^\top$.

\medskip
For $n\in\N^\star$, we denote by $0_n$ and $\mathrm{1}_n$ the vectors of zeros and ones in $\R^n$, by $\mathrm{I}_n$ the identity matrix in $\R^{n\times n}$, and by
\[
J_n\coloneqq \mathrm{1}_n\mathrm{1}_n^\top\in\R^{n\times n},
\]
the all-ones matrix. We write $(e_i)_{i\in\{1,\dots,n\}}$ for the canonical basis of $\R^n$, and $(E_{i,j})_{(i,j)\in\{1,\dots,n\}\times\{1,\dots,p\}}$ for the canonical basis of $\R^{n\times p}$, where $E_{i,j}$ has a single $1$ in position $(i,j)$ and $0$ elsewhere.

\medskip
For vectors $(x,y)\in\R^n\times\R^n$, we use the Euclidean inner product and norm
\[
x\cdot y\coloneqq x^\top y,
\;
\|x\|\coloneqq \sqrt{x^\top x}.
\]
For a matrix $A\in\R^{n\times p}$, we use the Frobenius norm
\[
\|A\|\coloneqq \sqrt{\tr[AA^\top]}=\sqrt{\sum_{i=1}^n\sum_{j=1}^p |A_{i,j}|^2}.
\]
If $W\in\R^{m\times m}$ is symmetric positive semidefinite and $u\in\R^m$, we write
\[
\|u\|_\smallertext{W}^2\coloneqq u^\top W u.
\]
Whenever we write inequalities such as $x\ge y$ for vectors or $A\ge B$ for matrices, they are understood \emph{componentwise}. By contrast, for symmetric matrices, $A\succeq 0$ (resp.\ $A\succ 0$) means that $A$ is positive semidefinite (resp.\ positive definite) in the quadratic-form sense. For a matrix $A$, we write $\ker(A)$ and $\mathrm{range}(A)$ for its kernel and range. If $A$ is symmetric, $\lambda_{\min}(A)$ denotes its smallest eigenvalue.

\medskip
For $x=(x_1,\dots,x_n)^\top\in\R^n$, we denote by
\[
\Diag[x]\coloneqq \Diag[x_1,\dots,x_n],
\]
the diagonal matrix with diagonal entries $x_1,\dots,x_n$. We use $\otimes$ for the Kronecker product and $\odot$ for the Hadamard (entrywise) product. The operator $\mathrm{vec}$ stacks the columns of a matrix: if $A\in\R^{n\times p}$, then $\mathrm{vec}(A)\in\R^{np}$ is obtained by concatenating the columns of $A$ from left to right.

\medskip
For stochastic processes, $\Ec(X)$ denotes the Dol\'eans--Dade stochastic exponential of a semimartingale $X$. 

\medskip
A certain number of constants will appear throughout the paper. We systematically assume
\[
T>0,
\;
s_0>0,
\;
\mu\in\R,
\;
\sigma>0,
\;
q_0\in\R^n,
\]
and, for each $i\in\{1,\dots,n\}$
\[
\nu_i>0,
\;
c_i>0,
\;
\gamma_i>0,
\;
\rho_i\in(-1,1).
\]
We also assume throughout that $\gamma_\smallertext{\rm P}>0$, with the limiting case $\gamma_\smallertext{\rm P} \longrightarrow 0$ corresponding to a risk-neutral principal.

\medskip
For derivatives, fix some $m\in\N^\star$. if $\phi:\R^m\longrightarrow\R$ is differentiable, then $D\phi(x)\in\R^m$ denotes its gradient at $x$, written as a column vector, and $D^2\phi(x)\in\R^{m\times m}$ denotes its Hessian matrix. We also use coordinate-wise notation $\partial_{x_\smalltext{i}}\phi$ and $\partial^2_{x_\smalltext{i} x_\smalltext{j}}\phi$. When the argument is a matrix variable, derivatives are taken entry-wise. For instance, if $\phi$ depends on a matrix $z\in\R^{n\times n}$, then
\[
\partial_{z}\phi
\coloneqq
\big(\partial_{z^{\smalltext{i}\smalltext{,}\smalltext{j}}}\phi\big)_{(i,j)\in\{1,\dots,n\}^\smalltext{2}},
\]
and mixed second derivatives are understood coordinate-wise in the same way. In particular, in the block first-order conditions of \Cref{prop:FOC-block}, rows of $z^{\smallertext Q}$ index the recipient contract and columns index the underlying signal: thus $z^{\smallertext{Q},i,j}$ is the loading of contract $i$ on signal $Q^j$, and the diagonal $z^{\smallertext{Q},i,i}$ is the only coefficient that enters agent $i$'s action. Finally, in the large-$\gamma_\smallertext{\rm P}$ analysis, we write $\spr(L_n)$ for the spectral radius of a matrix $L_n$, so as not to confuse it with the correlation coefficients $\rho$.}

\section{Principal--multiagent model}\label{sec:model}

We fix a positive integer $n$ and a probability space $(\Omega,\Fc,\P)$ carrying $2n$ independent and real-valued $\P$--Brownian motions $(W^i)_{i\in\{1,\dots,n\}}$ and $(B^i)_{i\in\{1,\dots,n\}}$. We let $\F^n\coloneqq (\Fc^n_t)_{t\in[0,T]}$ be their natural, $\P$-completed filtration, where $T>0$ is a fixed time horizon. $\mathcal T_{0,T}$ denotes the set of $\F^n$--stopping times valued in $[0,T]$.

\medskip

	We define, for $i\in\{1,\dots,n\}$, the processes
	\[
		Q^i_t\coloneqq q^i_0+\nu_i B^i_t,\; t\in[0,T],\;
		S_t=s_0+\int_0^t\mu S_s\mathrm{d}s+\int_0^t\sigma S_s \sum_{i=1}^n \frac{1}{\sqrt{n}} \Big( \rho_i \d B^i_s + \sqrt{1- \rho_i^2} \d W^i_s \Big),\; t\in[0,T].
	\]
and we denote 
\[
  q_0\coloneqq (q_0^1,\dots,q_0^n)^\top\in\R^n,
  \;
   Q\coloneqq (Q^1,\dots,Q^n)^\top.
   \] 
   Notice that the processes $Q^1,\dots,Q^n$ are driven by independent Brownian motions and therefore are not mutually correlated in the present model. What couples the agents is that each $Q^i$ is correlated with the traded factor $S$, through the coefficient $\rho_i$, and that the principal evaluates average compensation at the team level.

	\medskip
	
We consider a contracting problem where a principal designs a rewarding scheme for $n$ agents who produce \emph{distinct} ESG signals, each correlated with the traded factor $S$. Each agent $i\in\{1,\dots,n\}$ makes effort to disclose ESG data in continuous time, and $Q^i$ denotes the corresponding data-quality signal for agent $i$. The rewarding scheme designed by the principal will be a random vector $\xi = (\xi^1,\ldots,\xi^n)$, where each $\xi^i$ is $\Fc_T^n$-measurable and can be interpreted as a \emph{net token position} (equivalently, a net transfer claim) indexed to $S$.
	
	\medskip
	Now of course, the efforts or actions  of the agents impact the distribution of the data quality. 
	These actions are modelled as $\F^n$-adapted and measurable processes $(\alpha^i)_{i\in\{1,\dots,n\}}$, valued in $\R$ such that
	\[
	\E^\P\Bigg[\Ec\Bigg(\sum_{i=1}^n\int_0^\cdot\frac{\alpha^i_s}{\nu_i}\mathrm{d}B^i_s-\sum_{i=1}^n\int_0^{\cdot} \frac{\rho_i\alpha^i_s}{\nu_i \sqrt{1- \rho_i^2}}\mathrm{d}W^i_s\Bigg)_T\Bigg]=1.
	\]
	We denote this set by $\Ac_n$. Then, for any $\alpha\in\Ac_n$, the probability measure $\P^\alpha$ on $(\Omega, \Fc^n_T)$, whose Radon--Nikod\'ym density with respect to $\P$ is given by
	\[
	\frac{\mathrm{d}\P^\alpha}{\mathrm{d}\P}\coloneqq \Ec\Bigg(\sum_{i=1}^n\int_0^\cdot\frac{\alpha^i_s}{\nu_i}\mathrm{d}B^i_s-\sum_{i=1}^n\int_0^{\cdot}\frac{\rho_i\alpha^i_s}{\nu_i \sqrt{1- \rho_i^2}}\mathrm{d}W^i_s\Bigg)_T,
	\]
	is well-defined, and we have for $i\in\{1,\dots,n\}$
	\begin{equation} \label{eq:dynamicQS}
		Q^i_t= q^i_0+\int_0^t\alpha^i_s\mathrm{d}s+\nu_i B^{i,\alpha}_t,\; t\in[0,T],\;
		S_t=s_0+\int_0^t\mu S_s\mathrm{d}s+\int_0^t\sigma S_s
			\sum_{i=1}^n \frac{1}{\sqrt{n}} \Big( \rho_i \d B^{i,\alpha}_s + \sqrt{1- \rho_i^2} \d W^{i,\alpha}_s \Big),\; t\in[0,T],
	\end{equation}
where, by Girsanov's theorem, for any $i\in\{1,\dots,n\}$, $B^{i,\alpha}$ and $W^{i,\alpha}$ are two $\P^\alpha$-independent $(\F^n,\P^\alpha)$--Brownian motions given by
\[
B^{i,\alpha}_t\coloneqq B^i_t-\int_0^t\frac{\alpha^i_s}{\nu_i}\mathrm{d}s,\; t\in[0,T],\; W^{i,\alpha}_t\coloneqq W^i_t+\int_0^t\frac{\rho_i\alpha^i_s}{ \nu_i\sqrt{1- \rho_i^2}}\mathrm{d}s,\; t\in[0,T].
\]

	We notice from \eqref{eq:dynamicQS} that, under the probability measure $\P^{\alpha}$ with different effort process $\alpha \in \Ac$, the data quality process $Q$ follows different dynamics, 
	but the distribution of $S$ stays the same.
	
	\medskip
	Notice also the fact that the controls are $\R$-valued is deliberate: throughout the paper, $\alpha^i$ should be interpreted as a \emph{signed} disclosure / reporting-drift control. Positive values raise the drift of $Q^i$, while negative values capture quality shading or under-disclosure. This convention is important below when discussing possible negative diagonal loadings in the optimal contract.

\subsection{Problem of the agents}
Fix $i\in\{1,\dots,n\}$ and let $\alpha^{-i}=(\alpha^j)_{j\neq i}$ be an $\R^{n-1}$-valued,
$\F^n$-adapted and measurable process. We denote by $\Ac^i(\alpha^{-i})$ the set of $\R$-valued,
$\F^n$-adapted and measurable processes $\alpha$ such that $\alpha\oplus_i\alpha^{-i}\in\Ac_n$,
where
\[
\alpha\oplus_i\alpha^{-i}
\coloneqq 
\big(\alpha^1, \dots, \alpha^{i-1}, \alpha, \alpha^{i+1}, \dots, \alpha^n\big)^\top.
\]

Consider an arbitrary agent $i \in \{ 1, \dots,n\}$. Given the rewards $\xi\coloneqq(\xi^1,\dots,\xi^n)^\top$ chosen by the principal and given the actions  of the other agents represented by the $\R^{n-1}$-valued, $\F^n$-adapted and measurable process  $\alpha^{-i}$, agent $i$ aims to maximise his utility
	\begin{equation} \label{eq:pb_Agenti}
		V^i_0(\xi,\alpha^{-i}) \coloneqq \sup_{\alpha\in\Ac^\smalltext{i}(\alpha^{\smalltext{-}\smalltext{i}})} \E^{\P^{\smalltext{\alpha}\smalltext{\oplus}_\tinytext{i}\smalltext{\alpha}^{\tinytext{-}\tinytext{i}}}} \bigg[ U_i \bigg( \xi^i S_T - \frac{c_i}{2} \int_0^T \alpha^2_t \d t \bigg) \bigg],
	\end{equation}
	where
	\[
		U_i(x) \coloneqq - \mathrm{e}^{- \gamma_\smalltext{i} x},\; x\in\R,\; \text{for some given risk-aversion $\gamma_i>0$.}
	\]
We write $\alpha^\star \in \mathrm{NE}(\xi)$ and say that the joint action $\alpha^\star \in \Ac_n$ is a Nash equilibrium for the rewarding scheme $\xi$ if, for every $i\in\{1,\dots,n\}$, with
\[
\alpha^{\star,-i}= \big(\alpha^{\star,1}, \dots, \alpha^{\star,i-1}, \alpha^{\star,i+1}, \dots, \alpha^{\star,n}\big)^\top,
\]
one has
\[
V^i_0(\xi,\alpha^{\star,-i})
=
\E^{\P^{\alpha^\star}} \bigg[ U_i \bigg( \xi^i S_T - \frac{c_i}{2} \int_0^T (\alpha^{\star,i}_t)^2  \d t \bigg) \bigg]
=
\sup_{\alpha\in\Ac^i(\alpha^{\star,-i})}
\E^{\P^{\alpha\oplus_i\alpha^{\star,-i}}} \bigg[ U_i \bigg( \xi^i S_T - \frac{c_i}{2} \int_0^T \alpha_t^2  \d t \bigg) \bigg].
\]
Finally, we assume that each agent has a fixed exogenous reservation utility $R_i<0$ with corresponding certainty equivalent $r_i\coloneqq  - \frac{1}{\gamma_i} \ln( - R_i)$: the agents enter a contractual relationship with the principal only if $V^i_0(\xi,\alpha^{\star,-i} ) \ge R_i$ for every $i\in\{1,\dots,n\}$, for some $\alpha^\star \in \mathrm{NE}(\xi)\neq \emptyset$.

\subsection{Problem of the principal}
Following \citeauthor*{hernandez2024principal} \cite{hernandez2024principal}, we let $\Xi$ be the set of admissible contracts: in words, those are $\Fc_T^n$-measurable, $\R^n$-valued random vectors $\xi = (\xi^1,\dots,\xi^n)^\top$, with appropriate integrability (see \cite{hernandez2024principal} for details and below) and such that there exists an associated Nash equilibrium $\alpha^\star(\xi)\in\mathrm{NE}(\xi)\neq\emptyset$ with $V_0^i(\xi,\alpha^{\star,-i})\ge R_i$ for every $i\in\{1,\dots,n\}$. The principal then solves the optimisation problem
\[
V^P\coloneqq \sup_{\xi\in\Xi} \sup_{\alpha \in \mathrm{NE}(\xi)}
\E^{\P^{\alpha}}
\Big[ U_\smallertext{\rm P} \big( (Q_T\cdot\mathds{1}_n - \xi\cdot\mathds{1}_n S_T) /n \big ) \Big],
\]
where 
\[
U_\smallertext{\rm P}(x)\coloneqq - \mathrm{e}^{- \gamma_\smalltext{\rm P} x},\; x\in\R,\; \text{for some given risk-aversion $\gamma_\smallertext{\rm P}>0$,}
\]
while the case $\gamma_\smallertext{\rm P} \longrightarrow 0$ is understood throughout as the risk-neutral limiting regime of this exponential criterion.

\subsubsection{Reformulation of the problem of the principal}
\citeauthor*{hernandez2024principal} \cite{hernandez2024principal} provide a reformulation of the principal's problem as a standard stochastic control problem. In our linear--quadratic environment, this reformulation yields a closed-form solution. To state it, we introduce the following notation.

\medskip
First, define the spaces
\begin{gather*}
\mathbb V^Q\coloneqq \bigg\{ Z^\smallertext{Q}:\R^{n\times n}\text{-valued, $\F^n$-predictable processes},\;
\int_0^T \|Z_u^Q\|^2\mathrm{d}u < \infty,\; \P\text{--a.s.} \bigg\},\\
\mathbb V^S\coloneqq \bigg\{ Z^\smallertext{S}:\R^n\text{-valued, $\F^n$-predictable processes},\;
\int_0^T \|S_u Z_u^\smallertext{S}\|^2\mathrm{d}u < \infty,\; \P\text{--a.s.} \bigg\}.
\end{gather*}
Then, for each $i\in\{1,\dots,n\}$, let $F^i:\R^{n\times n}\times\R^n\longrightarrow\R$ be the function
\begin{align*}
F^i(z^\smallertext{Q},z^\smallertext{S})
&=
\frac{(z^{\smallertext{Q},i,i})^2}{2c_i}
+\sum_{j\in\{1,\dots,n\}\setminus\{i\}}\frac{z^{\smallertext{Q},j,j}}{c_j}z^{\smallertext{Q},i,j}
-\frac{\gamma_i}{2}\sum_{j=1}^n \nu_j^2 (z^{\smallertext{Q},i,j})^2
-\frac{\gamma_i}{2}\sigma^2 (z^{\smallertext{S},i})^2\\
&\quad
+\mu z^{\smallertext{S},i}
-\gamma_i\frac{\sigma}{\sqrt n}z^{\smallertext{S},i}\sum_{j=1}^n \rho_j\nu_j z^{\smallertext{Q},i,j}.
\end{align*}
For each $y\in\R^n$ and $Z\coloneqq(Z^\smallertext{Q},Z^\smallertext{S})\in\mathbb V^Q\times\mathbb V^S$, denote by $Y^{y,Z}$ the $\R^n$-valued process defined by
\begin{align*}
Y_t^{y,Z,i}
=
y^i
-\int_0^t F^i\big(Z_u^\smallertext{Q},S_u Z_u^\smallertext{S}\big)\mathrm{d}u
+\int_0^t Z_u^{\smallertext{Q},i,:}\cdot \mathrm{d}Q_u
+\int_0^t Z_u^{\smallertext{S},i}\mathrm{d}S_u,
\; i\in\{1,\dots,n\}.
\end{align*}
We let $\mathcal Z$ be the set of processes $Z\coloneqq(Z^\smallertext{Q},Z^\smallertext{S})\in \mathbb V^Q\times \mathbb V^S$ such that, in addition
\begin{itemize}
\item[$(i)$] $a^\star(Z)\in\Ac_n$, where $a^\star(Z)$ is the $\R^n$-valued, $\F^n$-adapted process defined by
\[
a^\star(Z)_t^i\coloneqq \frac{Z_t^{\smallertext{Q},i,i}}{c_i},
\; \mathrm{d}t\otimes \mathrm{d}\P\text{--a.e.},\; i\in\{1,\dots,n\};
\]
\item[$(ii)$] for all $i\in\{1,\dots,n\}$ and all $\alpha^i\in\Ac^i(a^\star(Z)^{-i})$, there exists $q_{\alpha^i}>1$ such that
\[
\sup_{\tau\in\mathcal T_{\smalltext{0}\smalltext{,}\smalltext{T}}}
\E^{\P^{\smalltext{\alpha}^\tinytext{i}\smalltext{\oplus}_\tinytext{i} \smalltext{a}^\tinytext{\star}\smalltext{(}\smalltext{Z}\smalltext{)}^{\tinytext{-}\tinytext{i}}}}
\Big[
\big|U_i(Y_\tau^{0,Z,i})\big|^{q_{\smalltext{\alpha}^\tinytext{i}}}
\Big]
<\infty
\]
\item[$(iii)$] the stochastic exponential $M^\smallertext{Z}$ defined below in \emph{\Cref{eq:MZ}} is an $(\F^n,\P^{a^\smalltext{\star}(Z)})$-martingale.
\end{itemize}

If $\xi\in\Xi$ is an admissible contract, then a Nash equilibrium $\alpha^\star\in\mathrm{NE}(\xi)$ can be characterised as
\[
\alpha_t^{\star,i}=a^\star(Z)_t^i=\frac{Z_t^{\smallertext{Q},i,i}}{c_i},
\; \mathrm{d}t\otimes\mathrm{d}\P\text{--a.e.},\; i\in\{1,\dots,n\},
\]
where $(Y,Z^\smallertext{Q},Z^\smallertext{S})$ is a solution to the $n$-dimensional BSDE
\begin{align*}
Y_t^i
=
\xi^i S_T
+\int_t^T F^i\big(Z_u^\smallertext{Q},S_uZ_u^\smallertext{S}\big)\mathrm{d}u
-\int_t^T Z_u^{\smallertext{Q},i,:}\cdot \mathrm{d}Q_u
-\int_t^T Z_u^{\smallertext{S},i}\mathrm{d}S_u,
\; i\in\{1,\dots,n\},
\end{align*}
such that $(Z^\smallertext{Q},Z^\smallertext{S})\in\mathcal Z$. Moreover
\[
U_i(Y_0^i)=V_0^i\big(\xi,a^\star(Z)^{-i}\big),
\; i\in\{1,\dots,n\}.
\]
We refer to \cite[Proposition 3.11]{hernandez2024principal} for proofs and details. Then the principal's optimal control problem becomes
	\begin{align*}
	V^P&= \sup_{y \ge r} \sup_{Z \in \mathcal Z} 
	\E^{\P^{a^\star(Z)}} \Big[
	 U_{\smallertext{\rm P}}\Big(  \mathds{1}_n \cdot \big( Q_T - Y^{y,\smallertext{Z}}_T\big) /n \Big)
	\Big],
	\end{align*}	
where $r \coloneqq (r_1, \dots,r_n)^\top$ and the inequality $y \ge r$ is component-wise.

\subsubsection{General optimal contract}	
Let $y \ge r$ and $Z=(Z^\smallertext{Q},Z^\smallertext{S})\in\mathcal Z$. Recall that $a^\star(Z)  \in \Ac_n$ and that the processes $B^{a^\smalltext{\star}(\smallertext{Z})}$ and $W^{a^\smalltext{\star}(\smallertext{Z})}$ defined by
		\[
B^{i,a^\smalltext{\star}(\smallertext{Z})}_t\coloneqq B^i_t-\int_0^t\frac{a^\star(Z_s)^i}{\nu_i}\mathrm{d}s,\; W^{i,a^\smalltext{\star}(\smallertext{Z})}_t\coloneqq W^i_t+\int_0^t\frac{\rho_ia^\star(Z_s)^i}{ \nu_i\sqrt{1- \rho_i^2}}\mathrm{d}s,\; t\in[0,T],\; i \in\{ 1, \dots, n\},
\]
 are two $\P^{a^\smalltext{\star}(\smallertext{Z})}$-independent $(\F^n,\P^{a^\smalltext{\star}(\smallertext{Z})})$--Brownian motions.  The dynamics of the processes $Q$ and $Y^{y,Z}$ satisfy
 	\begin{align*}
	Q^i_t &= q^i_0 + \int_0^t a^\star(Z_u)^i \d u + \int_0^t \nu_i \d B^{a^\smalltext{\star}(\smallertext{Z})}_u=
	q^i_0 + \int_0^t \dfrac{Z^{\smallertext{Q},i,i}_u}{c_i} \d u + \int_0^t \nu_i \d B^{a^\smalltext{\star}(\smallertext{Z})}_u\\
	Y^{y,\smallertext{Z},i}_t
	 &= y^i - \int_0^t \Big(F^i\big(Z^\smallertext{Q}_u,\overline{Z}^\smallertext{S}_u\big)  - Z^{\smallertext{Q},i,:} \cdot a^\star(Z)_u  - \mu \overline Z^{\smallertext{S},i}\Big)\d u + \int_0^t   \Sum_{j=1}^n \bigg( \nu_j Z^{\smallertext{Q},i,j}_u 
	+ \overline Z^{\smallertext{S},i}_u \dfrac{\sigma}{\sqrt{n} } \rho_j \bigg )
	\d B^{a^\smalltext{\star}(\smallertext{Z}),j}_u\\
	&\quad 
	+ \int_0^t \dfrac{\sigma}{\sqrt{n}} \overline Z^{\smallertext{S},i}_u \Sum_{j=1}^n  \sqrt{1- \rho_j^2} \d W^{a^\smalltext{\star}(\smallertext{Z}),j}_u,
	\end{align*}
where $\overline Z^\smallertext{S} \coloneqq  S Z^{\smallertext{S}}$. Then, we can write
	\begin{align*}
	v_{\smallertext{\rm P}}(y,Z)
	&\coloneqq \E^{\P^{a^\smalltext{\star}(\smallertext{Z})}} \bigg[
	 U_{\smallertext{\rm P}}\bigg( \frac{\mathds{1}_n \cdot \big( Q_T - Y^{y,\smallertext{Z}}_T\big)}{n} \bigg)
	\bigg] =
	- \mathrm{e}^{- \gamma_{\smalltext{\rm P}} \frac{\mathds{1}_\smalltext{n}\cdot (q_\smalltext{0} - y)}{n} }
	\E^{\P^{a^\smalltext{\star}(\smallertext{Z})}} 
	\bigg[
	 \exp \bigg(- \gamma_{\smallertext{\rm P}} \int_0^T f\big(Z^\smallertext{Q}_t, \overline{Z}^\smallertext{S}_t\big) \d t \bigg)
	 M^\smallertext{Z}_T
	\bigg].
	\end{align*}
where $M_T$ is the stochastic exponential
	\begin{align}\label{eq:MZ}
	M^\smallertext{Z}_T
	&\coloneqq  \Ec \Bigg(
	- \frac{\gamma_{\smallertext{\rm P}}}{n}
	\int_0^\cdot
	\sum_{i=1}^n \Bigg(
	\nu_i - { \nu_i \sum_{j=1}^n Z^{\smallertext{Q},j,i}_t }
	- \frac{\rho_i}{\sqrt{n}} \sum_{j=1}^n \sigma \overline{Z}^{\smallertext{S},j}_t
	\Bigg) \mathrm{d}B^{a^\smalltext{\star}(\smallertext{Z}),i}_t \notag\\
	&\quad
	- \frac{\gamma_{\smallertext{\rm P}}}{n}
	\int_0^\cdot \sum_{i=1}^n \Bigg(
	\frac{\sqrt{1-\rho_i^2}}{\sqrt{n}} \sum_{j=1}^n \sigma \overline{Z}^{\smallertext{S},j}_t
	\Bigg)  \d W^{a^\smalltext{\star}(\smallertext{Z}),i}_t
	\Bigg)_T,
	\end{align}
and 	$f$ is the function defined by
	\begin{align} \label{eq:fn_def}
		f(z^\smallertext{Q},z^\smallertext{S})
		&\coloneqq
		-\frac{1}{n} \sum_{i=1}^n
		\Bigg(
			 \frac{(z^{\smallertext{Q},i,i})^2}{2c_i} + \frac{\gamma_i}2
			 \sum_{j=1}^n\nu^2_j (z^{\smallertext{Q},i,j})^2
		+ \frac{ \gamma_i \sigma^2}2 \big(z^{\smallertext{S},i} \big)^2
			+\frac{\gamma_i\sigma}{\sqrt{n}} z^{\smallertext{S},i}\sum_{j=1}^n\rho_j \nu_j z^{\smallertext{Q},i,j}
			- \frac{z^{\smallertext{Q},i,i}}{c_i}
		\Bigg)\ \nonumber \\
		&
		-\frac{\gamma_{\smallertext{\rm P}}}{2}  \frac{1}{n^2} \sum_{i=1}^n
		\Bigg(
			\Bigg(
			\nu_i - { \nu_i \sum_{j=1}^n z^{\smallertext{Q},j,i} }
			- \frac{\rho_i}{\sqrt{n}} \sum_{j=1}^n \sigma z^{\smallertext{S},j}
		\Bigg)^2
			+
				\frac{(1-\rho_i^2)\sigma^2}{n} \Bigg(\sum_{j=1}^n z^{S,j}
			\Bigg)^2
		\Bigg).
	\end{align}
	
It follows that
	\begin{align*}
	V_\smallertext{\rm P} =  \sup_{y \ge r} \sup_{Z \in \mathcal Z}  v_\smallertext{\rm P}(y,Z)&=-
	\mathrm{e}^{- \gamma_{\smalltext{\rm P}} \frac{\mathds{1}_\smalltext{n}\cdot (q_\smalltext{0} - r)}{n} }
	\inf_{Z \in \mathcal Z}
	\E^{\P^{\smalltext{a}^\tinytext{\star}\smalltext{(}\smalltext{Z}\smalltext{)}}} 
	\Bigg[  \exp \bigg(- \gamma_{\smallertext{\rm P}} \int_0^T f\big(Z^\smallertext{Q}_t, \overline{Z}^\smallertext{S}_t\big) \d t \bigg) M^\smallertext{Z}_T		\Bigg].
	\end{align*}	
Using the martingale property of stochastic exponentials assumed in the definition of $\Zc$, and the fact that $M^Z$ is a true martingale for a deterministic $Z$, it is immediate that the optimisation in $V_\smallertext{\rm P}$ boils down to finding the deterministic maximiser of $f$. Notice that after some simple algebra, we can write
	\begin{equation}\label{eq:concavity-decomp}
\begin{aligned}
 f(z^\smallertext{Q},z^\smallertext{S})
&=
\frac{1}{2n}\sum_{i=1}^n\frac{1}{c_i}
-\frac{1}{2n}\sum_{i=1}^n \frac{(z^{\smallertext{Q},i,i}-1)^2}{c_i}
-\frac{1}{2n}\sum_{i=1}^n \gamma_i
\bigg(\sigma z^{\smallertext{S},i} + \frac{1}{\sqrt n}\rho^\top\big(\nu\odot (z^{\smallertext{Q},i,:})^\top\big)\bigg)^2
\\
&\quad
-\frac{1}{2n}\sum_{i=1}^n \gamma_i
\big(\nu\odot (z^{\smallertext{Q},i,:})^\top\big)^\top
\bigg(I_n-\frac{1}{n}\rho\rho^\top\bigg)
\big(\nu\odot (z^{\smallertext{Q},i,:})^\top\big)
\\
&\quad
-\frac{\gamma_{\smallertext{\rm P}}}{2n^2}\sum_{i=1}^n
\Bigg(
\Bigg(
\nu_i-\nu_i\sum_{j=1}^n z^{\smallertext{Q},j,i}
-\frac{\rho_i}{\sqrt n}\sum_{j=1}^n \sigma z^{\smallertext{S},j}
\Bigg)^2
+\frac{(1-\rho_i^2)\sigma^2}{n}
\Bigg(\sum_{j=1}^n z^{\smallertext{S},j}\Bigg)^2
\Bigg),
\end{aligned}
\end{equation}
where $\nu\coloneqq (\nu^1, \dots, \nu^n)^\top$ and $\rho\coloneqq (\rho^1, \dots, \rho^n)^\top$.

\paragraph{How to read the decomposition.}
The first negative square in \eqref{eq:concavity-decomp} is the \emph{row-by-row incentive cost}: it penalises deviations of the diagonal loading $z^{\smallertext Q,i,i}$ from the first-best benchmark $1$. The second and third lines are the \emph{risk-sharing terms}. The row-wise quadratic form captures the diffusion risk created when the contract loads on the non-traded signals $Q$, while the final principal term penalises the \emph{aggregate} exposure generated by column sums of $z^{\smallertext Q}$ and by the common $S$-tilt. In other words, $z^{\smallertext Q}$ creates incentives but also loads the contract on idiosyncratic and common noise, whereas $z^{\smallertext S}$ is the hedge used to re-balance that risk.

\medskip
Under this representation, $f$ is an affine--quadratic function of $(z^\smallertext{Q},z^\smallertext{S})$. Moreover, its Hessian is negative semidefinite because the matrix $\mathrm{I}_n-\frac{1}{n}\rho\rho^\top$ is symmetric positive semidefinite. Hence $f$ is globally concave. The proof of the following theorem is then straightforward once \Cref{prop:FOC-block} below is established. Notice also that since $s_0>0$ and $S$ is a geometric diffusion, one has $S_t>0$, $\P$--a.s.\ for all $t\in[0,T]$. In particular, both $S_T^{-1}$ and $\log(S_T/S_0)$ are well defined.

\begin{theorem}\label{th:main}
Let $(z^{\smallertext{S},\star},z^{\smallertext{Q},\star})$ be the unique maximiser of $f$.
Then the optimal contract is given by
\begin{align*}
\xi^i
&=S_T^{-1}\Bigg(
 r_i
-T\Bigg(\frac{(z^{\smallertext{Q},i,i,\star})^2}{2c_i}
 +\sum_{j\in\{1,\dots,n\}\setminus\{i\}}\frac{z^{\smallertext{Q},j,j,\star}}{c_j}z^{\smallertext{Q},i,j,\star}
 -\frac12\sum_{j=1}^n\gamma_i\nu_j^2 (z^{\smallertext{Q},i,j,\star})^2
 -\frac{\gamma_i}{2}\sigma^2 (z^{\smallertext{S},i,\star})^2\Bigg)\\
&\quad
-T\Bigg(\bigg(\mu-\frac{\sigma^2}{2}\bigg)z^{\smallertext{S},i,\star}
-\sum_{j=1}^n\gamma_i\frac{\rho_j}{\sqrt n}\nu_j z^{\smallertext{Q},i,j,\star}\sigma z^{\smallertext{S},i,\star}\Bigg)+z^{\smallertext{Q},i,: ,\star}\cdot(Q_T-Q_0)
+z^{\smallertext{S},i,\star}\log\bigg(\frac{S_T}{S_0}\bigg)
\Bigg),
\; i\in\{1,\dots,n\}.
\end{align*}
The explicit expressions of $(z^{\smallertext{Q},\star},z^{\smallertext{S},\star})$ are given by \eqref{eq:def_zSQ_star_n} below.
\end{theorem}

\paragraph{Interpretation of the optimal contract.}
The term $z^{\smallertext{Q},i,:,\star}\!\cdot (Q_T-Q_0)$ is the \emph{signal-based incentive leg}. Its diagonal coefficient $z^{\smallertext{Q},i,i,\star}$ rewards agent $i$'s own signal, while the off-diagonal entries $z^{\smallertext{Q},i,j,\star}$ implement cross-signal loadings across agents. The loading $z^{\smallertext{S},i,\star}\log(S_T/S_0)$ is the \emph{traded-asset hedge}: it offsets the diffusion risk created by the $Q$-exposures, but it also exposes the principal to aggregate market risk. The drift correction $\big(\mu-\sigma^2/2\big)z^{\smallertext{S},i,\star}$ is exactly the It\^o compensator associated with the log-return of $S$. Because signed exposures are allowed, the contract should be interpreted as a \emph{net transfer claim} indexed to $Q$ and $S$, rather than literally as a long-only token grant.

\subsubsection{First-order conditions}
In order to find the maximiser of $f$ in \eqref{eq:concavity-decomp}, let us look at the first-order conditions. 
For simplicity, let us denote by $(E_{i,j})_{(i,j)\in\{1,\dots,n\}^\smalltext{2}}$ the canonical basis of $\R^{n\times n}$, and define
\[
\lambda_n\coloneqq \frac{\gamma_{\smallertext{\rm P}}\sigma^2}{n^3}\sum_{j=1}^n(1-\rho_j^2)\in \R,\;
\Gamma\coloneqq \mathrm{Diag}[\gamma]\in\R^{n\times n},\;
N\coloneqq \mathrm{Diag}[\nu]\in\R^{n\times n},\;
\mathrm{J}\coloneqq \mathds{1}_n\mathds{1}_n^\top\in\R^{n\times n}.
\]
For any $i\in\{1,\dots,n\}$, set
\begin{gather*}
A_i\coloneqq \gamma_i+\frac{1}{c_i\nu_i^2},\;
\alpha_{i,n}\coloneqq \frac{\gamma_{\smallertext{\rm P}}}{n\gamma_i},\\
\kappa_{i,n}\coloneqq 1+\frac{\gamma_{\smallertext{\rm P}}}{n}\Bigg(\sum_{\ell=1}^n\frac1{\gamma_\ell}-\frac{1}{\gamma_i}\bigg(1-\frac{\gamma_i}{A_i}\bigg)\Bigg),\;
d_{i,n}\coloneqq \kappa_{i,n}^{-1}\bigg(\nu_i-\frac{1}{c_i\nu_iA_i}\bigg),\\
m_{i,n}\coloneqq \frac{\sigma}{\sqrt{n}\kappa_{i,n}}\rho_i\bigg(1-\frac{\gamma_i}{A_i}\bigg),\;
K_{i,n}(z^{\smallertext{S}})\coloneqq d_{i,n}-m_{i,n}z^{\smallertext{S},i},\; z^{\smallertext{S}}\in\R^n,\\
\mu_{i,n}\coloneqq \gamma_i\frac{\sigma^2}{n}\Bigg(1-\frac1n\sum_{j=1}^n\rho_j^2+\bigg(1-\frac{\gamma_i}{A_i}\bigg)\frac{\rho_i^2}{n}\Bigg)+ \frac{\gamma_{\smallertext{\rm P}}\sigma^2\rho_i^2}{\kappa_{i,n}n^3}\bigg(1-\frac{\gamma_i}{A_i}\bigg)^2,\\
D_n\coloneqq \mathrm{Diag}[\mu_{1,n},\dots,\mu_{n,n}],\\
\ell_{i,n}\coloneqq \frac{\gamma_{\smallertext{\rm P}}\sigma\rho_i}{n^{5/2}}\bigg(1-\frac{\gamma_i}{A_i}\bigg)d_{i,n}-\frac{\gamma_i\sigma}{n^{3/2}}\frac{\rho_i}{A_i c_i\nu_i},\;
\ell_n\coloneqq (\ell_{1,n},\dots,\ell_{n,n})^\top,\\
s_{i,n}\coloneqq \mu_{i,n}^{-1},\;
s_n\coloneqq (s_{1,n},\dots,s_{n,n})^\top,\;
S_n\coloneqq \mathrm{Diag}[s_n]=D_n^{-1},\;
y_n\coloneqq \frac{\lambda_n}{1+\lambda_n{1}_n^\top s_n}.
\end{gather*}

\begin{proposition}[First-order conditions]\label{prop:FOC-block}
The first-order conditions $\nabla f=0$ are the symmetric linear system
\begin{equation} \label{eq:FOC}
\begin{pmatrix}
H_{\smallertext{Q}\smallertext{Q}} & H_{\smallertext{Q}\smallertext{S}}\\
H_{\smallertext{Q}\smallertext{S}}^\top & H_{\smallertext{S}\smallertext{S}}
\end{pmatrix}
\begin{pmatrix}\mathrm{vec}(z^{\smallertext{Q}})\\ z^{\smallertext{S}}\end{pmatrix}
=
-\begin{pmatrix} b_\smallertext{Q}\\ b_\smallertext{S}\end{pmatrix},
\end{equation}
where $\mathrm{vec}$ stacks the columns of $z^{\smallertext{Q}}$, and with blocks 
\begin{align*}
H_{\smallertext{S}\smallertext{S}}
&\coloneqq -\frac{\sigma^2}{n}\Gamma - \frac{\gamma_{\smallertext{\rm P}}\sigma^2}{n^2}\mathrm{J},\\
H_{\smallertext{Q}\smallertext{S}}
&\coloneqq -\frac{\sigma}{n^{3/2}} (\rho\odot \nu)\otimes \Gamma 
	{ - \frac{\gamma_{\smallertext{\rm P}} \sigma}{n^{5/2}} (\rho\odot \nu)\otimes \mathrm{J}},\\
H_{\smallertext{Q}\smallertext{Q}}
&\coloneqq  -\frac{1}{n}\big( N^2 \otimes \Gamma\big)
	- \frac{\gamma_{\smallertext{\rm P}}}{n^2}\big(  N^2  \otimes \mathrm{J} \big) 
	-\frac{1}{n}\mathrm{Diag}\big[(1/c_1)E_{11},\dots,(1/c_n)E_{nn}\big],
\end{align*}
and right-hand side
\[
b_\smallertext{Q}^{i,j}\coloneqq\frac{1}{n c_i}\mathbf{1}_{\{i=j\}}+\frac{\gamma_{\smallertext{\rm P}}}{n^2}\nu_j^2,
\;
b_\smallertext{S}\coloneqq\frac{\gamma_{\smallertext{\rm P}}\sigma}{n^{5/2}}\Bigg(\sum_{i=1}^n \rho_i\nu_i\Bigg) \mathrm{1}_n.
\]
This system has a unique solution given by
\begin{equation} \label{eq:def_zSQ_star_n}
z^{\smallertext{S},\star}
\coloneqq
\big(S_n-y_n s_n s_n^\top\big)\ell_n,
\;
z^{\smallertext{Q},i,j,\star}\coloneqq \frac{q_{i,j}^\star}{\nu_j},\; (i,j)\in\{1,\dots,n\}^2,
\end{equation}
where, for any $(i,j)\in\{1,\dots,n\}^2$
\[
q_{i,j}^\star
\coloneqq
\begin{cases}
\displaystyle \alpha_{i,n} K_{j,n}(z^{\smallertext{S},\star}) - \frac{\sigma}{\sqrt n}\rho_j z^{\smallertext{S},i,\star}, \; \text{\rm if}\; j\neq i,\\[0.5em]
\displaystyle A_i^{-1}\bigg(\frac{\gamma_{\smallertext{\rm P}}}{n} K_{i,n}(z^{\smallertext{S},\star}) - \gamma_i\frac{\sigma}{\sqrt n}\rho_i z^{\smallertext{S},i,\star} + \frac{1}{c_i\nu_i}\bigg),\; \text{\rm if}\; j=i.
\end{cases}
\]
\end{proposition}

\section{Sign structure and comparative statics of the maximiser}\label{sec:homogeneous}

\subsection{Homogeneous case}
Assume that, for some constants $c$, $\gamma$, $\nu$ and $\rho$,
\begin{equation}\label{eq:assum_homog}
c_i\equiv c>0, 
\; \gamma_i\equiv\gamma>0, 
\; \nu_i\equiv\nu>0,
\; \rho_i\equiv\rho\in(-1,1).
\end{equation}
By symmetry, for each team size $n\ge 1$, the maximiser has the form, for $(i,j)\in\{1,\dots,n\}^2$,
\[
z_n^{\smallertext S,i,\star}=z_{s,n}^\star,
\;
z_n^{\smallertext Q,i,i,\star}=z_{d,n}^\star,
\;
z_n^{\smallertext Q,i,j,\star}=z_{o,n}^\star,\; \text{when }j\neq i.
\]
For later comparison across team sizes, we index by $n$ every scalar that depends explicitly on $n$. Set
\[
A\coloneqq \gamma+\frac{1}{c\nu^2},
\;
\delta\coloneqq 1-\frac{\gamma}{A}=\frac{A-\gamma}{A}=\frac{1}{Ac\nu^2},
\;
\alpha_n\coloneqq \frac{\gamma_\smallertext{\rm P}}{n\gamma},
\;
\beta_n\coloneqq \frac{\sigma\rho}{\sqrt n},
\]
\[
\widetilde\kappa_n\coloneqq A+\frac{\gamma_\smallertext{\rm P}}{n}\bigg(\frac{(n-1)A}{\gamma}+1\bigg),
\;
\kappa_n\coloneqq \frac{\widetilde\kappa_n}{A}.
\]
Finally define the scalar denominator
\begin{equation}
\label{eq:Delta-n}
\Delta_n(\gamma_\smallertext{\rm P})
\coloneqq
(\gamma+\gamma_\smallertext{\rm P})(1-\rho^2)
+\frac{\gamma\rho^2\delta}{n}
+\frac{\gamma_\smallertext{\rm P}\rho^2\delta^2}{n^2\kappa_n}
>0.
\end{equation}

\begin{proposition}[Homogeneous closed forms]\label{prop:homogeneous-closed-form}
Under the homogeneity assumptions in \eqref{eq:assum_homog}, the unique maximiser in {\rm\Cref{prop:FOC-block}} is given by
\begin{align}
\label{eq:zs-star}
z_{s,n}^\star
&=-\frac{\rho\gamma}{A c \sigma \nu \sqrt{n}\Delta_n(\gamma_\smallertext{\rm P})}
\bigg( 1-\frac{\gamma_\smallertext{\rm P}}{n\widetilde\kappa_n}\bigg),
\\[0.5em]
\label{eq:zozd-star}
z_{o,n}^\star
&=\frac{1}{\nu}\big(\alpha_n K_n^\star-\beta_n z_{s,n}^\star\big),
\; n\in\N^\star\setminus\{1\},
\\
z_{d,n}^\star
&=\frac{1}{\nu A}\bigg(\frac{\gamma_\smallertext{\rm P}}{n}K_n^\star-\gamma\beta_n z_{s,n}^\star+\frac{1}{c\nu}\bigg),
\end{align}
where 
\begin{equation}
\label{eq:K-star}
K_n^\star
\coloneqq
\frac{1}{\kappa_n}\bigg(\frac{\gamma\nu}{A}-\beta_n\delta z_{s,n}^\star\bigg).
\end{equation}
\end{proposition}

\paragraph{How to read the homogeneous closed forms.}
The constant $A=\gamma+\frac{1}{c\nu^2}$ is the basic \emph{own-signal incentive versus risk} trade-off: larger $c$ or smaller $\nu$ makes it harder to load on the diagonal signal. The quantity $\delta=1-\gamma/A$ is the residual share of that trade-off left after the own-signal term has been absorbed into $A$. The factors $\widetilde\kappa_n$ and $\kappa_n$ collect the feedback created by the principal's aggregate-risk penalty, while $\Delta_n(\gamma_\smallertext{\rm P})$ is the effective denominator obtained after eliminating the off-diagonal $Q$-weights and the column residual $K_n^\star$. The explicit factor $n^{-1/2}$ in $z_{s,n}^\star$ already shows that, under symmetry, the $S$-hedge is used to offset a common component whose contribution is averaged over the $n$ agents.

\begin{remark}[`n=1` sanity check and large-`n` scaling]\label{rem:homo-n1-large-n}
The off-diagonal quantity $z_{o,n}^\star$ is only meaningful for $n\ge 2$. When $n=1$, the formulas reduce to
\begin{gather*}
z_{s,1}^\star
=-\frac{\gamma\nu\rho}{\sigma(\gamma+\gamma_\smallertext{\rm P})\big(1+c\nu^2(\gamma+\gamma_\smallertext{\rm P})(1-\rho^2)\big)},\\
z_{d,1}^\star
=\frac{1+\gamma_\smallertext{\rm P}c\nu^2(1-\rho^2)}{1+(\gamma+\gamma_\smallertext{\rm P})c\nu^2(1-\rho^2)}.
\end{gather*}
Thus the diagonal formula in {\rm\Cref{prop:homogeneous-closed-form}} is consistent with the one-agent benchmark: when $n=1$, there is no peer-benchmarking channel and the only remaining trade-off is between the own-signal incentive and the traded-asset hedge.

\medskip
For fixed $(c,\gamma,\nu,\rho,\sigma,\gamma_\smallertext{\rm P})$ and $n\longrightarrow\infty$, one directly reads from the closed forms that
\[
z_{s,n}^\star=\Oc(n^{-1/2}),
\;
z_{o,n}^\star=\Oc(n^{-1}),
\;
z_{d,n}^\star\longrightarrow \frac{1}{1+c\gamma\nu^2}.
\]
In particular, as the team becomes large, the common $S$-tilt and each off-diagonal benchmarking term become negligible, while the diagonal coefficient converges to the own-signal loading that would arise from the row-wise problem after aggregate exposure has been diversified away.
\end{remark}

\begin{proposition}[Signs, limits, and basic comparative statics]\label{prop:homo-signs-limits}
Under the homogeneity condition \eqref{eq:assum_homog}, the closed forms \eqref{eq:zs-star}--\eqref{eq:zozd-star} satisfy
\begin{itemize}
\item[$(i)$] $\mathrm{sgn}(z_{s,n}^\star)=-\mathrm{sgn}(\rho);$

\medskip
\item[$(ii)$] $z_{o,n}^\star\ge0$ for $n\in \N^\star\setminus\{1\}$, with strict inequality except in the degenerate case $(\gamma_\smallertext{\rm P},\rho)\longrightarrow (0,0)$, and $z_{d,n}^\star>0$ for all $n\in\N^\star;$

\medskip
\item[$(iii)$] $|z_{s,n}^\star|$ is decreasing in $\gamma_\smallertext{\rm P}$ and $z_{s,n}^\star\longrightarrow 0$ as $\gamma_\smallertext{\rm P}\longrightarrow\infty$;

\medskip
\item[$(iv)$] for $n\in\N^\star\setminus\{1\}$
\[
\lim_{\gamma_\smalltext{\rm P}\to\infty} z_{o,n}^\star
=\frac{\gamma}{(n-1)A+\gamma}\in(0,1);
\]

\medskip
\item[$(v)$] for all $n\in\N^\star$
\[
\lim_{\gamma_\smalltext{\rm P}\to\infty} z_{d,n}^\star
=\frac{(n-1)A-(n-2)\gamma}{(n-1)A+\gamma}\in(0,1].
\]
\end{itemize}
\end{proposition}

\begin{proof}
$(i)$ The denominator \eqref{eq:Delta-n} satisfies $\Delta_n(\gamma_\smallertext{\rm P})>0$, and
\[
0\le \frac{\gamma_\smallertext{\rm P}}{n\widetilde\kappa_n}
=
\frac{\gamma_\smallertext{\rm P}}
{nA+\gamma_\smallertext{\rm P}\big(\frac{(n-1)A}{\gamma}+1\big)}
<1,
\; \gamma_\smallertext{\rm P} > 0,
\]
because the denominator is strictly larger than $\gamma_\smallertext{\rm P}$. Hence $\mathrm{sgn}(z_{s,n}^\star)=-\mathrm{sgn}(\rho)$.

\medskip
$(ii)$ From \eqref{eq:K-star}
\[
K_n^\star=\frac{1}{\kappa_n}\bigg(\frac{\gamma\nu}{A}-\beta_n\delta z_{s,n}^\star\bigg)\ge \frac{\gamma\nu}{A\kappa_n}>0,
\]
since $\beta_n z_{s,n}^\star\le 0$. Therefore
\[
z_{o,n}^\star=\frac{1}{\nu}(\alpha_nK_n^\star-\beta_n z_{s,n}^\star)\ge0.
\]
If $\gamma_\smallertext{\rm P}>0$, then $\alpha_n>0$ and $K_n^\star>0$, so $z_{o,n}^\star>0$. When $\gamma_\smallertext{\rm P}\longrightarrow0$, then
\[
z_{o,n}^\star=-\frac{1}{\nu}\beta_n z_{s,n}^\star,
\]
which is strictly positive when $\rho\neq0$ and equal to $0$ when $\rho=0$. Finally
\[
z_{d,n}^\star=\frac{1}{\nu A}\bigg(\frac{\gamma_\smallertext{\rm P}}{n}K_n^\star-\gamma\beta_n z_{s,n}^\star+\frac{1}{c\nu}\bigg)>0.
\]

\medskip
$(iii)$ In \eqref{eq:zs-star}, the numerator factor $1-\gamma_\smallertext{\rm P}/(n\widetilde\kappa_n)$ decreases with $\gamma_\smallertext{\rm P}$, while the denominator \eqref{eq:Delta-n} increases with $\gamma_\smallertext{\rm P}$. Hence $|z_{s,n}^\star|$ decreases with $\gamma_\smallertext{\rm P}$, and $z_{s,n}^\star\longrightarrow 0$ as $\gamma_\smallertext{\rm P}\longrightarrow\infty$.

\medskip
$(iv)$ Using $z_{s,n}^\star\longrightarrow 0$ and \eqref{eq:K-star}
\[
\lim_{\gamma_\smalltext{\rm P}\to\infty} z_{o,n}^\star
=\lim_{\gamma_\smalltext{\rm P}\to\infty}\frac{1}{\nu}\alpha_n K_n^\star
=\lim_{\gamma_\smalltext{\rm P}\to\infty}\frac{\gamma_\smalltext{\rm P}}{n\widetilde\kappa_n}
=\frac{\gamma}{(n-1)A+\gamma}.
\]

\medskip
$(v)$ The limit for $z_{d,n}^\star$ follows similarly from \eqref{eq:zozd-star} and \eqref{eq:K-star} by direct algebra.
\end{proof}

We next isolate the results for a risk-neutral principal. Define
\[
\eta_n\coloneqq \rho^2\bigg((n-1)+\frac{\gamma}{A}\bigg)\in[0,n).
\]
At the limit case $\gamma_\smallertext{\rm P}\longrightarrow0$, \Cref{prop:FOC-block} decouples across rows. The homogeneous maximiser admits the following closed forms.

\begin{proposition}[Homogeneous formulas when $\gamma_\smallertext{\rm P} \to 0$]
\label{prop:homog-gp0}
When $\gamma_\smallertext{\rm P}\longrightarrow0$, the unique maximiser is
\begin{align}
\label{eq:zs-gp0}
z_{s,n}^{0}
&\coloneqq z_{s,n}^\star\big|_{\gamma_\smalltext{\rm P}=0}
= -\frac{\sqrt{n}}{\sigma}\frac{\rho}{A c \nu}\frac{1}{n-\eta_n},
\\[0.25em]
\label{eq:zo-gp0}
z_{o,n}^{0}
&\coloneqq z_{o,n}^\star\big|_{\gamma_\smalltext{\rm P}=0}
= \frac{\rho^2}{A c \nu^2}\frac{1}{n-\eta_n},
\; n\in\N^\star\setminus\{1\},
\\[0.25em]
\label{eq:zd-gp0}
z_{d,n}^{0}
&\coloneqq z_{d,n}^\star\big|_{\gamma_\smalltext{\rm P}=0}
= \frac{1}{A c \nu^2}\bigg(1+\frac{\gamma\rho^2/A}{n-\eta_n}\bigg).
\end{align}
Moreover, with
\[
\nu_{n,\rho}^{\dagger}
\coloneqq
\sqrt{\frac{n(1-\rho^2)+\rho^2}{\gamma cn(1-\rho^2)}},
\]
one has
\begin{itemize}
\item $|z_{s,n}^0|$ is increasing in $|\rho|$ and decreasing in $c;$
\item $|z_{s,n}^0|$ is increasing in $\nu$ on $(0,\nu_{n,\rho}^{\dagger}]$ and decreasing in $\nu$ on $[\nu_{n,\rho}^{\dagger},\infty);$
\item $z_{o,n}^0$ increases with $|\rho|$ and decreases with both $c$ and $\nu;$
\item $z_{d,n}^0$ increases with $|\rho|$ and decreases with both $c$ and $\nu$.
\end{itemize}
\end{proposition}

\begin{proof}
The closed forms are obtained by taking the limit $\gamma_\smallertext{\rm P}\longrightarrow0$ in \Cref{prop:homogeneous-closed-form}. For the comparative statics, it is convenient to rewrite
\begin{gather*}
|z_{s,n}^0|
=
\frac{\sqrt n|\rho|\,\nu}{\sigma\big(c\gamma n(1-\rho^2)\nu^2+n(1-\rho^2)+\rho^2\big)},\\
z_{o,n}^0
=
\frac{\rho^2}{c\gamma n(1-\rho^2)\nu^2+n(1-\rho^2)+\rho^2},\\
z_{d,n}^0
=
\frac{n(1-\rho^2)+\rho^2}{c\gamma n(1-\rho^2)\nu^2+n(1-\rho^2)+\rho^2}.
\end{gather*}
The stated monotonicities in $|\rho|$, $c$, and $\nu$ now follow by direct differentiation. In particular,
\[
\partial_{\nu}|z_{s,n}^0|
\;\text{has the sign of}\;
n(1-\rho^2)+\rho^2-\gamma c\,n(1-\rho^2)\nu^2,
\]
which yields the threshold $\nu_{n,\rho}^{\dagger}$.
\end{proof}

In particular, \Cref{prop:homog-gp0} can be read as the exact $\gamma_\smallertext{\rm P}=0$ simplification of the homogeneous closed forms, while \Cref{prop:hetero-small-gp-persistence} describes the stability of the associated sign pattern for small positive principal risk aversion.

\begin{figure}[ht!]
  \centering
  \includegraphics[width=.32\linewidth]{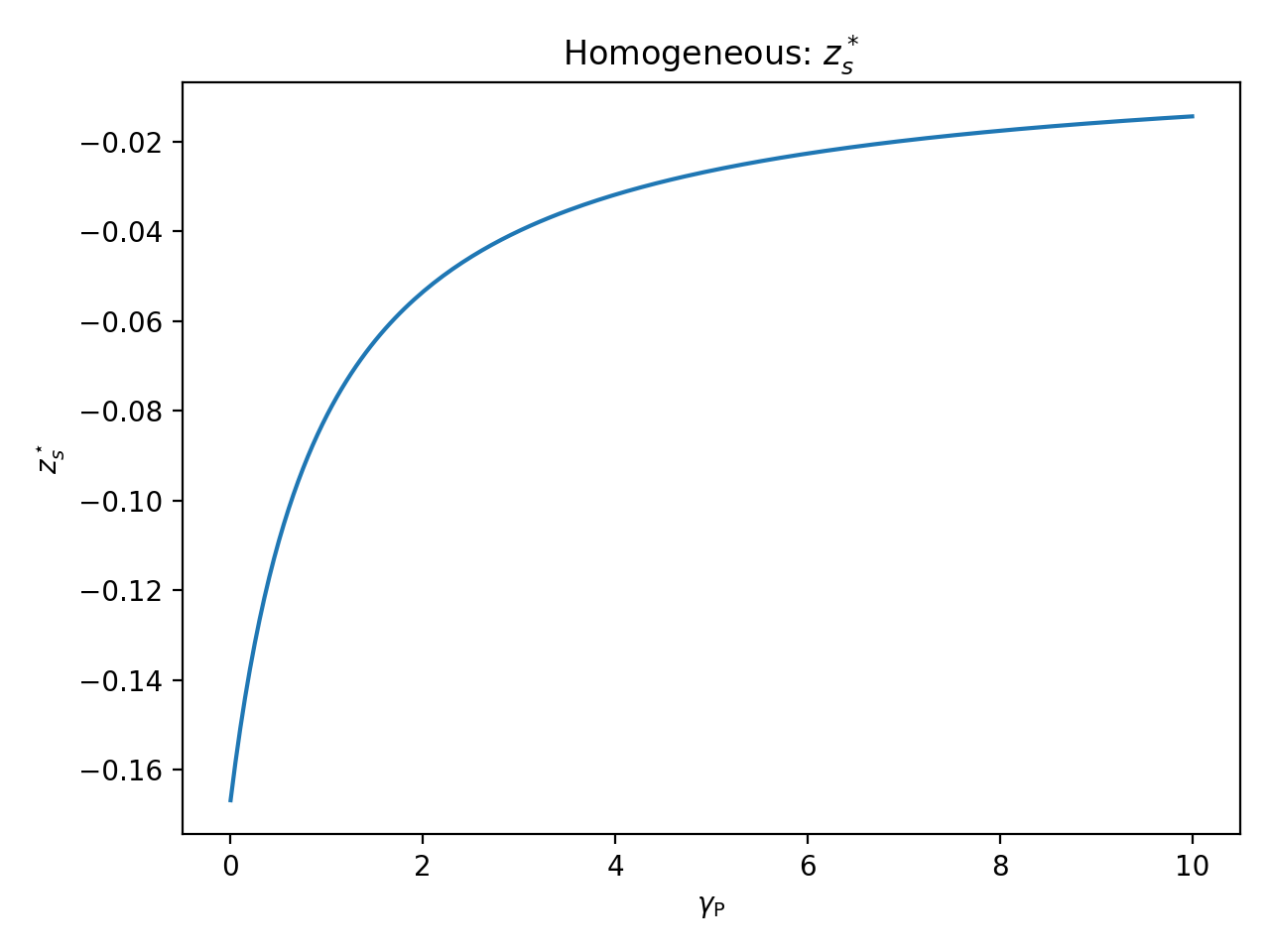}\hfill
  \includegraphics[width=.32\linewidth]{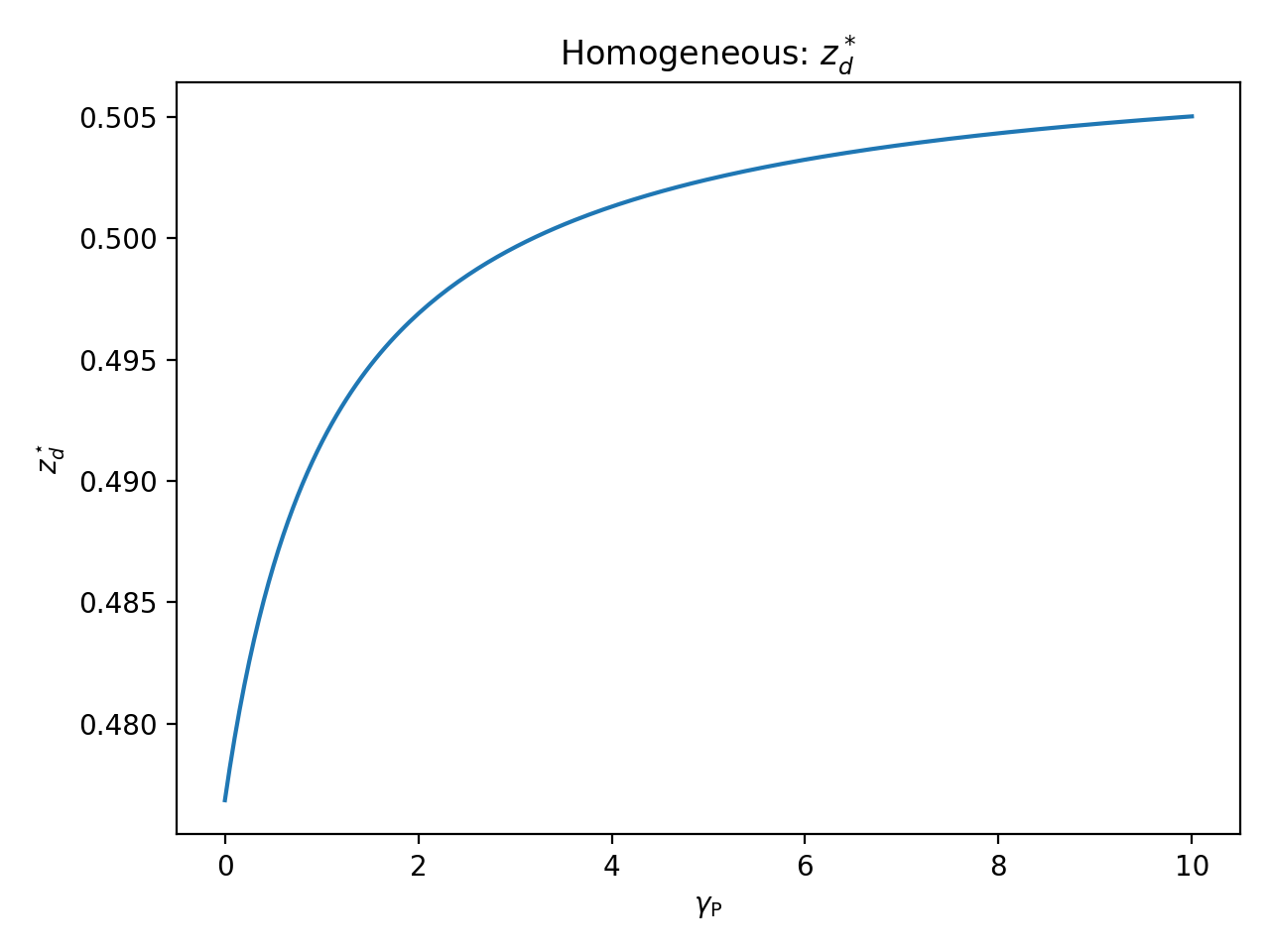}\hfill
  \includegraphics[width=.32\linewidth]{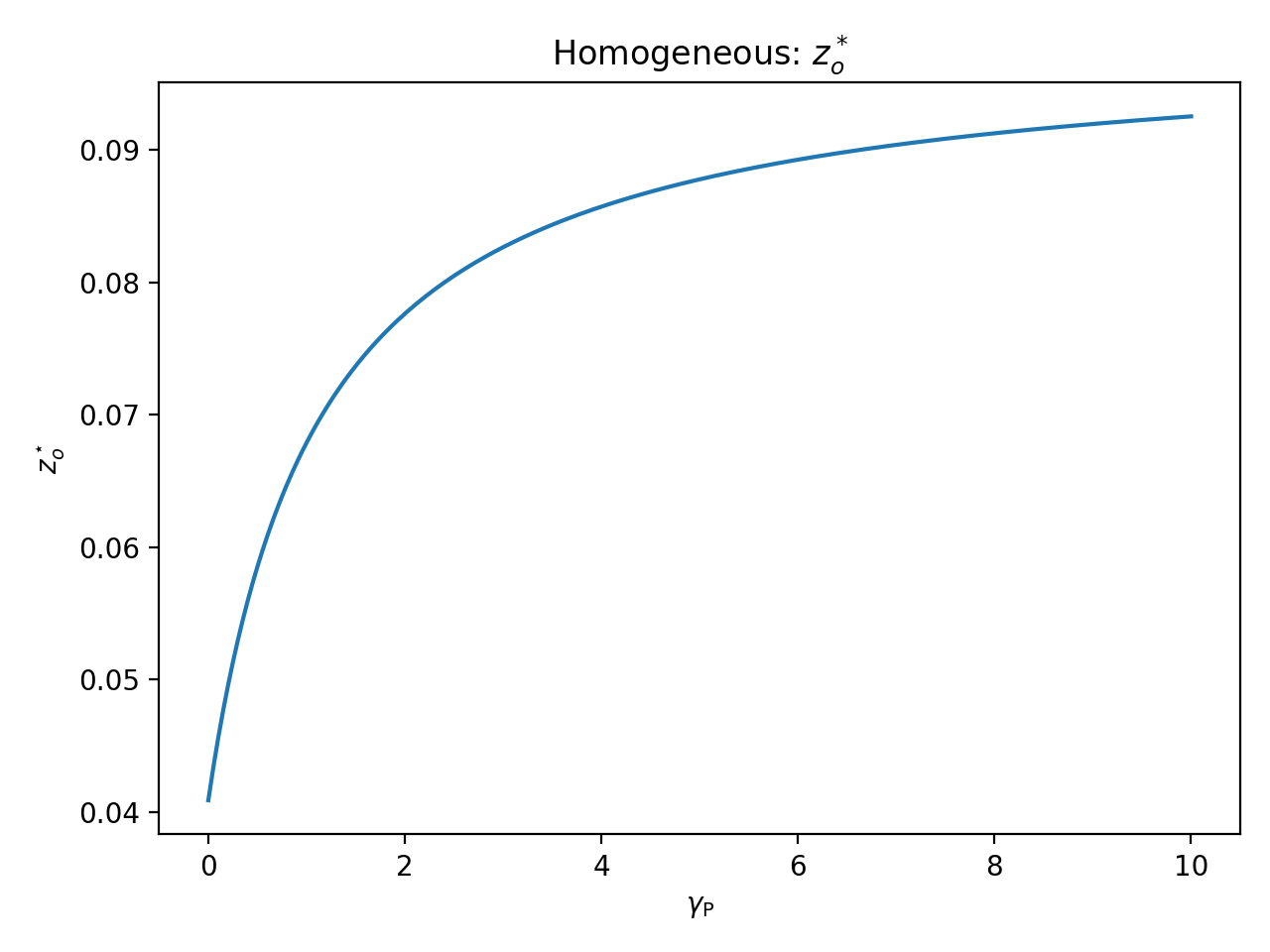}
  \caption{\small Homogeneous economy. Closed forms from \Cref{prop:homogeneous-closed-form}.}
  \label{fig:homo-all}
\end{figure}

\subsection{Heterogeneous case in the large risk aversion limit}
\label{subsec:large-gP-explicit}

We study in this subsection the limit case $\gamma_{\smallertext{\rm P}}\longrightarrow\infty$. Recall from \eqref{eq:fn_def} that
\[
f(z^\smallertext{Q},z^\smallertext{S})
=
g(z^\smallertext{Q},z^\smallertext{S})
-\frac{\gamma_{\smallertext{\rm P}}}{2}\Phi_n(z^\smallertext{Q},z^\smallertext{S}),
\]
with
\begin{align*}
g(z^\smallertext{Q},z^\smallertext{S})
&\coloneqq
-\frac{1}{n} \sum_{i=1}^n \Bigg(
\frac{(z^{\smallertext{Q},i,i})^2}{2c_i}
+\frac{\gamma_i}{2}\sum_{j=1}^n \nu_j^2 (z^{\smallertext{Q},i,j})^2
+\frac{\gamma_i\sigma^2}{2} (z^{\smallertext{S},i})^2
+\frac{\gamma_i\sigma}{\sqrt{n}}z^{\smallertext{S},i}\sum_{j=1}^n \rho_j\nu_j z^{\smallertext{Q},i,j}
\Bigg)
+\frac{1}{n}\sum_{i=1}^n \frac{z^{\smallertext{Q},i,i}}{c_i},
\end{align*}
and
\begin{align*}
\Phi_n(z^\smallertext{Q},z^\smallertext{S})
&\coloneqq \frac{1}{n^2} \sum_{i=1}^n \Bigg(
\bigg(\nu_i-\nu_i\sum_{j=1}^n z^{\smallertext{Q},j,i}
- \frac{\rho_i\sigma}{\sqrt{n}} \sum_{k=1}^n z^{\smallertext{S},k}\bigg)^2
+\frac{(1-\rho_i^2)\sigma^2}{n}\bigg(\sum_{k=1}^n z^{\smallertext{S},k}\bigg)^2\Bigg) \\
&=\big\|O_nx-o_n\big\|_{\smallertext W_n}^2
 \coloneqq (O_nx-o_n)^\top W_n (O_nx-o_n),
\end{align*}
where $W_n\in\R^{(n+1)\times(n+1)}$ is the diagonal matrix
\[
W_n\coloneqq \mathrm{Diag}\bigg[\underbrace{\frac{1}{n^2},\dots,\frac{1}{n^2}}_{\text{$n$ times}}, \frac{\sigma^2}{n^3}\sum_{i=1}^n(1-\rho_i^2)\bigg],
\]
and $O_n:\R^{n^2+n}\longrightarrow \R^{n+1}$ is the linear operator
\[
\begin{aligned}
(O_nx)_i&\coloneqq \nu_i\sum_{j=1}^n z^{\smallertext{Q},j,i} + \frac{\rho_i\sigma}{\sqrt{n}}\sum_{k=1}^n z^{\smallertext{S},k},\; i\in\{1,\dots,n\},\; (O_nx)_{n+1}&\coloneqq \sum_{k=1}^n z^{\smallertext{S},k},
\; x=(z^\smallertext{Q},z^\smallertext{S}),
\end{aligned}
\]
and
\[
o_n\coloneqq (\nu_1,\dots,\nu_n,0)^\top\in\R^{n+1}.
\]

\paragraph{Penalty interpretation.}
The quantity $\Phi_n(z^\smallertext Q,z^\smallertext S)$ is the squared weighted distance to the affine constraint set $\frF_n$ defined below in \Cref{eq:Fn}. When $\gamma_{\smallertext{\rm P}}$ is large, the principal first forces the contract close to these constraints---column sums of $z^\smallertext Q$ close to $1$, and aggregate $S$-exposure close to $0$---and only then optimises the residual row-by-row risk-sharing problem inside that affine set. This is the sense in which the large-$\gamma_{\smallertext{\rm P}}$ regime isolates the pure risk-management component of the contract.

\medskip
When $\gamma_{\smallertext{\rm P}}\to\infty$, the second term acts as a quadratic penalty. Thus we expect the optimisation problem
\begin{equation} \label{eq:opitmize_f_gamma}
\max_{x=(z^\smallertext{Q},z^\smallertext{S})\in\R^{\smalltext{n}^\tinytext{2}\smalltext{+}\smalltext{n}}}
f_\gamma(x)\coloneqq g(x)-\frac{\gamma_{\smallertext{\rm P}}}{2}\|O_nx-o_n\|_{\smallertext W_n}^2,
\end{equation}
to converge to the constrained problem
\begin{equation} \label{eq:opitmize_g_F}
\max_{x\in \R^{\smalltext{n}^\tinytext{2}\smalltext{+}\smalltext{n}}}  g(x)
\;\text{subject to}\;
x=(z^\smallertext{Q},z^\smallertext{S})\in\frF_n,
\end{equation}
where
\begin{equation}\label{eq:Fn}
\frF_n\coloneqq \big\{x\in\R^{n^\smalltext{2}+n}:O_nx=o_n\big\}
=
\Bigg\{(z^\smallertext{Q},z^\smallertext{S})\in\R^{n\times n}\times\R^n:
\sum_{j=1}^n z^{\smallertext{Q},j,i}=1,\;  i\in\{1,\dots,n\},\; 
\sum_{k=1}^n z^{\smallertext{S},k}=0\Bigg\}.
\end{equation}

\begin{lemma}[Existence, uniqueness, and KKT for the constrained problem]\label{lem:KKT-constrained}
For the constrained optimisation problem \eqref{eq:opitmize_g_F}, there exists a unique solution $x_n^\star =  \big(\mathrm{vec}(z_n^{\smallertext{Q},\star}),z_n^{\smallertext{S},\star}\big) \in \frF_n$.
Moreover, there is a unique multiplier $\iota_n^\star\in\R^{n+1}$ such that the KKT system
\begin{equation}\label{eq:KKT-constrained}
\begin{cases}
\partial g(x_n^\star)+O_n^\top \iota_n^\star=0,\\
O_nx_n^\star=o_n,
\end{cases}
\end{equation}
holds.
Equivalently,
\[
H_nx_n^\star+h_n+O_n^\top\iota_n^\star=0,\; O_nx_n^\star=o_n,
\]
where $H_n\coloneqq D^2 g$ and $h_n\coloneqq \partial g(0)$.
\end{lemma}

\begin{proof}
Notice first that the function $g$ is a strictly concave quadratic, so there exists $m_n>0$ such that
\begin{equation}\label{eq:strong-concavity}
-y^\top H_n y\ge m_n\|y\|^2,
\; y\in\R^{n^2+n},
\end{equation}
Then, since $\frF_n$ is non-empty and closed, there exists a unique maximiser $x_n^\star$ of $g$ over $\frF_n$. By \Cref{lem:row-rank}, $O_n$ has full row rank, hence the multiplier $\iota_n^\star$ is unique. The KKT conditions are necessary and sufficient for equality-constrained concave maximisation.
\end{proof}

For the penalised problem \eqref{eq:opitmize_f_gamma}, \Cref{prop:FOC-block} yields a unique global maximiser $x_{\gamma,n}$ for each $\gamma_{\smallertext{\rm P}}>0$. Moreover, it satisfies
\begin{equation}\label{eq:stat-penalty}
D g(x_{\gamma,n})-\gamma_{\smallertext{\rm P}}O_n^\top W_n(O_nx_{\gamma,n}-o_n)=0
\Longleftrightarrow
H_nx_{\gamma,n}+h_n=\gamma_{\smallertext{\rm P}}O_n^\top W_n r_{\gamma,n},
\end{equation}
where
\[
r_{\gamma,n}\coloneqq O_nx_{\gamma,n}-o_n.
\]
The next result formalises the convergence from the unconstrained to the constrained problem.
\begin{proposition}\label{thm:penalty-limit}
Let $x_n^\star$ solve the constrained problem \eqref{eq:KKT-constrained}, and let $x_{\gamma,n}$ maximise $f_\gamma$. Define the scaled multipliers
\[
\widehat\iota_{\gamma,n}\coloneqq -\gamma_{\smallertext{\rm P}}W_n r_{\gamma,n}\in\R^{n+1}.
\]
Then, as $\gamma_{\smallertext{\rm P}}\to\infty$,
\begin{itemize}
\item[$(i)$] $\|O_nx_{\gamma,n}-o_n\|_{\smallertext W_n}=\Oc(1/\gamma_{\smallertext{\rm P}})$;

\item[$(ii)$] $x_{\gamma,n}\to x_n^\star$, and $\|x_{\gamma,n}-x_n^\star\|=\Oc(1/\gamma_{\smallertext{\rm P}})$;

\item[$(iii)$] $\widehat\iota_{\gamma,n}\to \iota_n^\star$, and $\|\widehat\iota_{\gamma,n}-\iota_n^\star\|=\Oc(1/\gamma_{\smallertext{\rm P}})$;

\item[$(iv)$] $\big|g(x_n^\star)-g(x_{\gamma,n})\big|=\Oc(1/\gamma_{\smallertext{\rm P}})$.
\end{itemize}
\end{proposition}

We now solve the KKT system explicitly. For $i\in\{1,\dots,n\}$ set
\[
\mathcal H_{i,n}\coloneqq \mathrm{Diag}\bigg[\gamma_i\nu_1^2,\dots,\gamma_i\nu_{i-1}^2, \gamma_i\nu_i^2+\frac1{c_i}, \gamma_i\nu_{i+1}^2,\dots,\gamma_i\nu_n^2\bigg],
\;
\mathcal P_{i,n}\coloneqq \mathcal H_{i,n}^{-1}=\mathrm{Diag}[p_{i,1},\dots,p_{i,n}],
\]
so that $p_{i,j}=1/(\gamma_i\nu_j^2)$ if $j\neq i$ and $p_{i,i}=1/(\gamma_i\nu_i^2+1/c_i)$. Define, for each column $j\in\{1,\dots,n\}$,
\[
\Theta_{j,n}\coloneqq \sum_{i=1}^n p_{i,j},
\;
\zeta_j\coloneqq \frac{p_{j,j}}{c_j}\in(0,1),
\;
a_j\coloneqq \frac{\rho_j}{c_j\nu_j(\gamma_j\nu_j^2+1/c_j)}
=\frac{\rho_j\zeta_j}{\nu_j}.
\]
For each row $k\in\{1,\dots,n\}$ set
\begin{gather*}
M_{k,j}\coloneqq \rho_j\nu_j p_{k,j},
\;
\upsilon_{k,n}\coloneqq \sum_{j=1}^n \rho_j^2\nu_j^2 p_{k,j},
\;
\varphi_{k,n}\coloneqq 1-\frac{\gamma_k}{n}\upsilon_{k,n}\in(0,1),\\
w_{j,n}\coloneqq \frac{1}{n\Theta_{j,n}},
\;
\varphi_n\coloneqq (\varphi_{1,n},\dots,\varphi_{n,n})^\top.
\end{gather*}
Finally define the $n\times n$ matrix $L_n$ and the vectors $(U_n,V_n)\in\R^n\times\R^n$ by
\begin{gather*}
(L_n)_{k,j}\coloneqq \frac{1}{\varphi_{k,n}}w_{j,n}a_jM_{k,j}
=\frac{\rho_j^2\zeta_j p_{k,j}}{n\varphi_{k,n}\Theta_{j,n}}\ge 0,\; (k,j)\in\{1,\dots,n\}^2,
\;
(V_n)_k\coloneqq \frac{n}{\gamma_k\sigma^2\varphi_{k,n}},
\\
(U_n)_k\coloneqq -\frac{1}{\varphi_{k,n}}\Bigg(
\underbrace{\frac{1}{\sigma\sqrt n}\sum_{j=1}^n \frac{1-\zeta_j}{\Theta_{j,n}}M_{k,j}}_{\eqqcolon C_{\smalltext{k}\smalltext{,}\smalltext{n}}}
+\underbrace{\frac{1}{\sigma\sqrt n}\frac{\rho_k\nu_k}{c_k}p_{k,k}}_{\eqqcolon r_{\smalltext{k}\smalltext{,}\smalltext{n}}}
\Bigg),
\; (k,j)\in\{1,\dots,n\}^2.
\end{gather*}

\begin{proposition}[Explicit constrained maximiser on $\frF_n$]\label{prop:explicit-column}
Set
\[
u_n\coloneqq (\mathrm{I}_n-L_n)^{-1}V_n,
\;
v_n\coloneqq (\mathrm{I}_n-L_n)^{-1}U_n,
\;
\vartheta_n^\star\coloneqq -\frac{{1}_n^\top v_n}{{1}_n^\top u_n}.
\]
Then the unique maximiser of $g$ on $\frF_n$ is given by
\begin{gather}
\bar z_n^{\smallertext S}
\coloneqq u_n\vartheta_n^\star+v_n,\\
\bar z_n^{\smallertext Q,i,j}
\coloneqq p_{i,j}\bigg(n\bar\mu_{j,n}+\frac{1}{c_i}\mathbf{1}_{\{i=j\}}-\frac{\gamma_i\sigma}{\sqrt{n}}\bar z_n^{\smallertext S,i}\rho_j\nu_j\bigg),
\; (i,j)\in\{1,\dots,n\}^2,
\label{eq:zQ-column}
\end{gather}
where
\begin{equation}
\bar\mu_{j,n}\coloneqq\frac{1}{n\Theta_{j,n}}\bigg(1-\zeta_j-\frac{\sigma}{\sqrt{n}}a_j\bar z_n^{\smallertext S,j}\bigg),
\; j\in\{1,\dots,n\}.
\label{eq:mu-column}
\end{equation}
Moreover $L_n\ge 0$ and $L_n^\top\varphi_n<\varphi_n.$ Consequently $\spr(L_n)<1$, so $(\mathrm{I}_n-L_n)$ is a nonsingular $M$-matrix and
\[
(\mathrm{I}_n-L_n)^{-1}\ge 0.
\]
\end{proposition}

\paragraph{Interpretation of the reduced system.}
The multipliers $\bar\mu_{j,n}$ are the column-by-column corrections that enforce $\sum_i \bar z_n^{\smallertext Q,i,j}=1$. Once those have been eliminated, the only genuine degrees of freedom are the entries of $\bar z_n^{\smallertext S}$ subject to $\mathrm{1}_n^\top \bar z_n^{\smallertext S}=0$. The matrix $L_n$ measures how the choice of the $S$-tilt in column $j$ feeds back into row $k$ through the column constraints, so the formula $(\mathrm{I}_n-L_n)^{-1}$ is the exact resolvent of that correlation-rebalancing mechanism.

\begin{remark}[Structural simplification as $\gamma_\smallertext{\rm P}\to\infty$]
At finite $\gamma_\smallertext{\rm P}>0$, the first-order conditions form a coupled $(n^2+n)\times(n^2+n)$ linear system. In the penalty limit, the constraints $O_nx=o_n$ bind and the problem reduces to
\begin{enumerate}
\item[$(i)$] $n$ decoupled diagonal solves, one per row $i$, delivering
\[
q_i
=
\mathcal P_{i,n}\bigg(
\frac{1}{c_i}e_i
-\frac{\gamma_i\sigma}{\sqrt n}\bar z_n^{\smallertext S,i}(\rho\odot\nu)
+n(\bar\mu_{1,n},\dots,\bar\mu_{n,n})^\top
\bigg);
\]
\item[$(ii)$] a single $n\times n$ reduced system of nonsingular $M$-matrix type, solved in closed form through $(\mathrm{I}_n-L_n)^{-1}\ge0$ and the linear constraint $\mathrm{1}_n^\top \bar z_n^{\smallertext S}=0$.
\end{enumerate}
\end{remark}

\begin{proposition}[Sign structure when $\rho$ is one-sided]\label{prop:signs-large-gp-column}
Suppose $\rho_i$ share a common sign $s\in\{-1,1\}$ and are not all zero.
Let $u_n$ and $v_n$ be as in {\rm\Cref{prop:explicit-column}}. Then $u_n>0$ entrywise, while $v_n$ has the constant sign $-s$.
Consequently, the constrained-limit vector $\bar z_n^{\smallertext S}=u_n\vartheta_n^\star+v_n$ satisfies exactly one of the following two alternatives:
\begin{itemize}
\item either $\bar z_n^{\smallertext S}\equiv 0$;
\item or $\bar z_n^{\smallertext S}$ has mixed signs.
\end{itemize}
In particular, every non-zero constrained limit has at least one positive and one negative component.
\end{proposition}

\begin{proof}
Since $(\mathrm{I}_n-L_n)^{-1}\ge 0$ and $V_n>0$, one has
\[
u_n=(\mathrm{I}_n-L_n)^{-1}V_n>0.
\]
If $s=1$, then every $M_{k,j}\ge 0$, so $C_{k,n}>0$ and $r_{k,n}>0$ for every $k$; hence $(U_n)_k<0$ and therefore
\[
v_n=(\mathrm{I}_n-L_n)^{-1}U_n<0.
\]
If $s=-1$, the same argument with reversed signs gives $v_n>0$.
Thus $v_n$ has the constant sign $-s$. Now
\[
\vartheta_n^\star=-\frac{\mathrm{1}_n^\top v_n}{\mathrm{1}_n^\top u_n},
\]
so $\vartheta_n^\star$ has sign $s$. Since $\mathrm{1}_n^\top \bar z_n^{\smallertext S}=0$, the vector $\bar z_n^{\smallertext S}$ cannot be entrywise non-negative unless it is identically zero, and it cannot be entrywise non-positive unless it is identically zero. Therefore either $\bar z_n^{\smallertext S}\equiv 0$, or it has both a positive and a negative component.
\end{proof}

\begin{remark}
If $c_i\equiv c$, $\gamma_i\equiv\gamma$, $\nu_i\equiv\nu$, and $\rho_i\equiv\rho$, then symmetry gives $\bar z_n^{\smallertext S}\equiv 0$, and \eqref{eq:zQ-column} reduces to the positive, symmetric $Q$-weights consistent with the homogeneous formulas proved in {\rm\Cref{prop:homogeneous-closed-form}}. This is precisely the exceptional zero case allowed by {\rm\Cref{prop:signs-large-gp-column}}.
\end{remark}

The next results concentrate on sign properties for the optimal sensitivity to $Q$.

\begin{proposition}[Diagonal sign test in the large--$\gamma_{\rm P}$ regime]
\label{prop:diag-sign-large-gp}
Let $(\bar z_n^{\smallertext Q},\bar z_n^{\smallertext S})$ be the unique maximiser of $g$ on $\frF_n$.

\medskip
$(i)$ For each $i\in\{1,\dots,n\}$,
\[
\mathrm{sgn}\big(\bar z_n^{\smallertext S,i}\big)=-\mathrm{sgn}(\rho_i)
\Longrightarrow
\bar z_n^{\smallertext Q,i,i}>0.
\]
Consequently, there is no diagonal sign flip in the limit as long as row $i$ keeps the standard $S$-tilt opposite to $\rho_i$.

\medskip
$(ii)$ For each $i\in\{1,\dots,n\}$,
\[
\bar z_n^{\smallertext Q,i,i}
= p_{i,i}\Bigg(
\underbrace{\frac{1-\zeta_i}{\Theta_{i,n}}+\frac{1}{c_i}}_{\text{\rm positive baseline}}
 - \underbrace{\frac{\sigma}{\sqrt n}\bar z_n^{\smallertext S,i}\rho_i
\bigg(\gamma_i\nu_i+\frac{p_{i,i}}{c_i\nu_i\Theta_{i,n}}\bigg)}_{\text{\rm correlation-balancing correction}}
\Bigg).
\]
Hence
\[
\mathrm{sgn}\big(\bar z_n^{\smallertext Q,i,i}\big)=\mathrm{sgn}\big(\mathscr B_{i,n}\big),\;
\mathscr B_{i,n}\coloneqq
\frac{1-\zeta_i}{\Theta_{i,n}}+\frac{1}{c_i}
-\frac{\sigma}{\sqrt n}\bar z_n^{\smallertext S,i}\rho_i
\bigg(\gamma_i\nu_i+\frac{p_{i,i}}{c_i\nu_i\Theta_{i,n}}\bigg).
\]
In particular, if $\rho_i\neq 0$, $\mathrm{sgn}(\bar z_n^{\smallertext S,i})=\mathrm{sgn}(\rho_i)$, and
\[
|\bar z_n^{\smallertext S,i}| >
\frac{\sqrt{n}}{\sigma|\rho_i|
\Big(\gamma_i\nu_i+\frac{p_{i,i}}{c_i\nu_i\Theta_{i,n}}\Big)}
\bigg(\frac{1-\zeta_i}{\Theta_{i,n}}+\frac{1}{c_i}\bigg),
\]
then $\bar z_n^{\smallertext Q,i,i}<0$. If $\rho_i=0$, then $a_i=0$ and the diagonal term is automatically positive. 

\medskip
$(iii)$ For $(i,j)\in\{1,\dots,n\}^2$ with $i\neq j$,
\[
\bar z_n^{\smallertext Q,i,j}
= p_{i,j}\Bigg(
\underbrace{\frac{1-\zeta_j}{\Theta_{j,n}}}_{>0}
 - \underbrace{\frac{\gamma_i\sigma}{\sqrt n}\bar z_n^{\smallertext S,i}\rho_j\nu_j
+\frac{\sigma}{\sqrt n}\frac{a_j}{\Theta_{j,n}}\bar z_n^{\smallertext S,j}}_{\text{\rm correlation rebalancing}}
\Bigg).
\]
Thus $\bar z_n^{\smallertext Q,i,j}>0$ whenever
\[
\gamma_i\big|\bar z_n^{\smallertext S,i}\big||\rho_j|\nu_j
+\frac{|a_j|}{\Theta_{j,n}}\big|\bar z_n^{\smallertext S,j}\big|
\le
\frac{1-\zeta_j}{\Theta_{j,n}}\frac{\sqrt n}{\sigma}.
\]
In particular, if $|\bar z_n^{\smallertext S}|$ is moderate and $\rho$ is weakly dispersed, all off-diagonals remain positive.
\end{proposition}

\begin{proof}
$(i)$ From the explicit diagonal entry,
\[
\bar z_n^{\smallertext Q,i,i}
= p_{i,i}\bigg(
n\bar\mu_{i,n}+\frac{1}{c_i}-\frac{\gamma_i\sigma}{\sqrt{n}}\bar z_n^{\smallertext S,i}\rho_i\nu_i
\bigg).
\]
Insert the formula for $\bar\mu_{i,n}$:
\[
n\bar\mu_{i,n}+\frac{1}{c_i}
= \frac{1}{\Theta_{i,n}}(1-\zeta_i)+\frac{1}{c_i}-\frac{\sigma}{\sqrt{n}}\frac{a_i}{\Theta_{i,n}}\bar z_n^{\smallertext S,i}.
\]
Hence
\[
\bar z_n^{\smallertext Q,i,i}
= p_{i,i}\mathscr B_{i,n},
\;
\mathscr B_{i,n}
=\underbrace{\bigg(\frac{1}{\Theta_{i,n}}(1-\zeta_i)+\frac{1}{c_i}\bigg)}_{C_{i,n}^0>0}
-\frac{\sigma}{\sqrt{n}}\underbrace{\bigg(\frac{a_i}{\Theta_{i,n}}+\gamma_i\rho_i\nu_i\bigg)}_{R_{i,n}\ \text{has sign }\mathrm{sgn}(\rho_i)}\bar z_n^{\smallertext S,i}.
\]
If $\mathrm{sgn}(\bar z_n^{\smallertext S,i})=-\mathrm{sgn}(\rho_i)$, then $R_{i,n}\bar z_n^{\smallertext S,i}<0$, so $\mathscr B_{i,n}>C_{i,n}^0>0$. Therefore $\bar z_n^{\smallertext Q,i,i}>0$.

\medskip
$(ii)$ The displayed decomposition is obtained by evaluating \eqref{eq:zQ-column} at $j=i$ and substituting \eqref{eq:mu-column}. Since $p_{i,i}>0$, the sign is the sign of $\mathscr B_{i,n}$. The sufficient condition is exactly the negativity condition for $\mathscr B_{i,n}$ under $\mathrm{sgn}(\bar z_n^{\smallertext S,i})=\mathrm{sgn}(\rho_i)$. When $\rho_i=0$, one has $a_i=0$, so the correction term vanishes and
\[
\mathscr B_{i,n}=\frac{1-\zeta_i}{\Theta_{i,n}}+\frac{1}{c_i}>0.
\]
Hence $\bar z_n^{\smallertext Q,i,i}>0$ automatically in that case.

\medskip
$(iii)$ Insert \eqref{eq:mu-column} into \eqref{eq:zQ-column} with $j\ne i$:
\[
\bar z_n^{\smallertext Q,i,j}
= p_{i,j}\bigg(
\frac{1-\zeta_j}{\Theta_{j,n}}
-\frac{\gamma_i\sigma}{\sqrt n}\bar z_n^{\smallertext S,i}\rho_j\nu_j
-\frac{\sigma}{\sqrt n}\frac{a_j}{\Theta_{j,n}}\bar z_n^{\smallertext S,j}
\bigg).
\]
The sufficient positivity condition follows by bounding the last two terms in absolute value.
\end{proof}

\subsection{Heterogeneous case for small \texorpdfstring{$\gamma_{\smallertext{\rm P}}>0$}{γP>0}}
\label{subsec:signs-gp0-small}

At the limit case $\gamma_\smallertext{\rm P}\longrightarrow0$, the rows decouple. For $i\in\{1,\dots,n\}$ set
\[
\pi_{i,n}\coloneqq \|\rho\|^2-\rho_i^2+\frac{\gamma_i}{A_i}\rho_i^2\in[0,n).
\]

\begin{proposition}[Closed forms and signs when $\gamma_\smallertext{\rm P} \to 0$]
\label{prop:hetero-gp0}
Let
\[
q_{i,j,n}^0\coloneqq \nu_j z^{\smallertext{Q},i,j,\star}\big|_{\gamma_\smalltext{\rm P}=0},
\;
s_{i,n}^0\coloneqq z^{\smallertext{S},i,\star}\big|_{\gamma_\smalltext{\rm P}=0}.
\]
Then, for all $(i,j)\in\{1,\dots,n\}^2$,
\begin{align}
q_{i,j,n}^0&= - \frac{\sigma}{\sqrt{n}} s_{i,n}^0\rho_j, \; j\ne i,
\label{eq:hetero-gp0-qoff}\\
q_{i,i,n}^0&= \frac{1}{A_i}\bigg(\frac{1}{c_i\nu_i}-\frac{\gamma_i\sigma}{\sqrt{n}}\rho_i s_{i,n}^0\bigg),
\label{eq:hetero-gp0-qdiag}\\
s_{i,n}^0&= -\frac{\sqrt{n}}{\sigma}\frac{ \rho_i}{ A_i c_i\nu_i(n-\pi_{i,n})}.
\label{eq:hetero-gp0-s}
\end{align}
In particular, $\mathrm{sgn}(s_{i,n}^0)= -\mathrm{sgn}(\rho_i),\;
z^{\smallertext{Q},i,i,\star}(0)>0,\; i\in\{1,\dots,n\},$ and, for $j\ne i$, $\mathrm{sgn}\big(z^{\smallertext{Q},i,j,\star}(0)\big)=\mathrm{sgn}(\rho_i\rho_j).$ \end{proposition}

\begin{proof}
When $\gamma_\smallertext{\rm P}\longrightarrow0$, the objective $f$ separates across rows $i$ and is strictly concave; the row-$i$ first-order conditions give a $2\times 2$ linear system in $\big(q_{i,\cdot},s_i\big)$ whose unique solution is \eqref{eq:hetero-gp0-qoff}--\eqref{eq:hetero-gp0-s}. The sign conclusions are immediate from \eqref{eq:hetero-gp0-s} and the formulas above.
\end{proof}

\begin{proposition}[Comparative statics when $\gamma_\smallertext{\rm P} \to 0$]
\label{prop:hetero-gp0-CS}
From \eqref{eq:hetero-gp0-s}--\eqref{eq:hetero-gp0-qdiag}, for any $(i,j)\in\{1,\dots,n\}^2$
\begin{itemize}
\item $|s_{i,n}^0|$ increases with $|\rho_i|$ and decreases with $c_i;$
\item with
\[
\nu_{i,n}^{\dagger}
\coloneqq
\sqrt{\frac{n-\|\rho\|^2+\rho_i^2}{\gamma_i c_i(n-\|\rho\|^2)}},
\]
$|s_{i,n}^0|$ is increasing in $\nu_i$ on $(0,\nu_{i,n}^{\dagger}]$ and decreasing in $\nu_i$ on $[\nu_{i,n}^{\dagger},\infty);$
\item for $j\ne i$, $|q_{i,j,n}^0|=\frac{\sigma}{\sqrt n}|s_{i,n}^0||\rho_j|$ increases with $|\rho_j|$ and with $|s_{i,n}^0|;$
\item $q_{i,i,n}^0$ increases with $|\rho_i|$.
\end{itemize}
\end{proposition}

\begin{proof}
Using \eqref{eq:hetero-gp0-s}, one can rewrite
\[
|s_{i,n}^0|
=
\frac{\sqrt n|\rho_i|\,\nu_i}{\sigma\big(\gamma_i c_i(n-\|\rho\|^2)\nu_i^2+n-\|\rho\|^2+\rho_i^2\big)}.
\]
Hence $|s_{i,n}^0|$ is increasing in $|\rho_i|$ and decreasing in $c_i$ by direct differentiation, and
\[
\partial_{\nu_i}|s_{i,n}^0|
\;\text{has the sign of}\;
n-\|\rho\|^2+\rho_i^2-\gamma_i c_i(n-\|\rho\|^2)\nu_i^2,
\]
which gives the threshold $\nu_{i,n}^{\dagger}$. The off-diagonal formula
\[
|q_{i,j,n}^0|=\frac{\sigma}{\sqrt n}|s_{i,n}^0||\rho_j|,
\]
immediately yields the third claim. Finally,
\[
q_{i,i,n}^0
=
\frac{\nu_i(n-\|\rho\|^2+\rho_i^2)}{\gamma_i c_i(n-\|\rho\|^2)\nu_i^2+n-\|\rho\|^2+\rho_i^2},
\]
so $q_{i,i,n}^0$ increases with $|\rho_i|$.
\end{proof}

We next look at what we can deduce for small $\gamma_\smallertext{\rm P}>0$.

\begin{proposition}[Local sign persistence for small $\gamma_\smallertext{\rm P}>0$]
\label{prop:hetero-small-gp-persistence}
Fix $n\in\N^\star$. There exists $\overline\gamma_{\smallertext{\rm P},n}>0$ such that, for every $\gamma_\smallertext{\rm P}\in(0,\overline\gamma_{\smallertext{\rm P},n})$,
\[
z^{\smallertext Q,i,i,\star}(\gamma_\smallertext{\rm P})>0,
\; i\in\{1,\dots,n\}.
\]
Moreover, for every $i\in\{1,\dots,n\}$ such that $\rho_i\neq0$
\[
\mathrm{sgn}\big(z^{\smallertext S,i,\star}(\gamma_\smallertext{\rm P})\big)
=
-\mathrm{sgn}(\rho_i),
\]
and for every $(i,j)\in\{1,\dots,n\}^2$ with $j\neq i$ and $\rho_i\rho_j\neq0$
\[
\mathrm{sgn}\big(z^{\smallertext Q,i,j,\star}(\gamma_\smallertext{\rm P})\big)
=
\mathrm{sgn}(\rho_i\rho_j).
\]
If $\rho_i=0$ or $\rho_i\rho_j=0$, continuity only implies that the corresponding coefficient remains small for small $\gamma_\smallertext{\rm P}$; it need not remain exactly zero.
\end{proposition}

\begin{proof}
By \Cref{prop:FOC-block}, the maximiser is the unique solution of a nonsingular linear system whose coefficient matrix and right-hand side depend smoothly on $\gamma_\smallertext{\rm P}$. Therefore the solution map
\[
\gamma_\smallertext{\rm P}
\longmapsto
\big(z^{\smallertext Q,\star}(\gamma_\smallertext{\rm P}),z^{\smallertext S,\star}(\gamma_\smallertext{\rm P})\big),
\]
is continuous at $0$.

\medskip
By \Cref{prop:hetero-gp0}, the diagonal entries $z^{\smallertext Q,i,i,\star}(0)$ are strictly positive for all $i$. The sign of $z^{\smallertext S,i,\star}(0)$ is strict whenever $\rho_i\neq0$, and the sign of $z^{\smallertext Q,i,j,\star}(0)$ for $j\neq i$ is strict whenever $\rho_i\rho_j\neq0$. Hence these strict sign relations persist on some interval $[0,\overline\gamma_{\smallertext{\rm P},n})$ for a suitable $\overline\gamma_{\smallertext{\rm P},n}>0$.
\end{proof}

\subsection{Economic interpretation of optimal sensitivities}
\label{subsec:econ-interpretation}

The sensitivity vectors $z^{\smallertext S,\star}$ and the matrix $z^{\smallertext Q,\star}$ jointly determine how the optimal contract reacts to the traded asset $S$ and to the signals $Q$. The key trade‑off is: \emph{incentives on $Q$} raise variance via diffusion risk, while \emph{tilts on $S$} hedge that variance but may expose the principal to aggregate risk. The parameter $\gamma_{\smallertext{\rm P}}$ gauges how much the principal penalises aggregate exposure.

\paragraph{Regime A: risk–neutral principal (the limit case $\gamma_\smallertext{\rm P}\to 0$).}
By \Cref{prop:hetero-gp0}, rows decouple exactly and the solution has a transparent structure
\[
\mathrm{sgn}\big(z^{\smallertext S,i,\star}\big)=-\mathrm{sgn}(\rho_i),\;
z^{\smallertext Q,i,i,\star}>0,\; i\in\{1,\dots,n\},\; 
\mathrm{sgn}\big(z^{\smallertext Q,i,j,\star}\big)=\mathrm{sgn}(\rho_i\rho_j),\; (i,j)\in\{1,\dots,n\}^2,\; j\ne i.
\]
\emph{Interpretation.} 
$(i)$ Each $S$–tilt hedges the variance created by row‑$i$ incentives, hence it leans \emph{against} $\rho_i$. 

\medskip
$(ii)$ Off–diagonals co‑move with the sign pattern of $(\rho_i\rho_j)$ so that the $Q$-leg and the $S$-hedge reinforce each other, lowering the risk cost per unit of incentive. 

\medskip
$(iii)$ From \Cref{prop:hetero-gp0-CS}, stronger absolute correlations $|\rho|$ push the solution to rely more on $S$ and to increase the magnitude of the cross-signal loadings, while higher $c_i$ (costlier effort) dampens the reliance on $S$.

\paragraph{Regime B: slightly risk--averse principal (small $\gamma_\smallertext{\rm P}>0$).}
The block system in \Cref{prop:FOC-block} adds rank-one penalties to $H_{\smallertext{S}\smallertext{S}}$ and $H_{\smallertext{Q}\smallertext{Q}}$, shrinking only the common components. By \Cref{prop:hetero-small-gp-persistence}, there exists an $n$-dependent threshold $\overline\gamma_{\smallertext{\rm P},n}>0$ such that every sign that is \emph{strict} when $\gamma_\smallertext{\rm P}\longrightarrow0$ persists throughout $[0,\overline\gamma_{\smallertext{\rm P},n})$.

\medskip
\emph{Interpretation.} Relative to the risk-neutral benchmark, the principal subtracts a common component from the vector of $S$-tilts and pushes each column sum $\sum_{j=1}^n z^{\smallertext Q,j,i}$ towards $1$ (equivalently, $\nu_i\sum_{j=1}^n z^{\smallertext Q,j,i}$ towards $\nu_i$), while preserving locally the non-degenerate sign pattern inherited from $\gamma_\smallertext{\rm P}\longrightarrow0$.

\paragraph{Regime C: highly risk–averse principal (large $\gamma_{\rm P}$).}
The quadratic penalty enforces the affine constraints in the limit (\Cref{thm:penalty-limit}), yielding the explicit constrained maximiser in \Cref{prop:explicit-column}
\[
\sum_{k=1}^n \bar z_n^{\smallertext S,k}=0,\;
\sum_{i=1}^n \bar z_n^{\smallertext Q,i,j}=1,\; j\in\{1,\dots,n\}.
\]
\emph{Interpretation.} The principal eliminates \emph{aggregate} $S$ exposure (market‑neutrality) and reshapes $Q$-weights so that each column aggregates to unity (identity pooling), while still sharing diffusion risk via heterogeneous $S$-tilts and correlation‑aware off-diagonals.

\medskip
\emph{Sign structure of $S$-tilts.} In general heterogeneity, $\bar z_n^{\smallertext S}$ solves a reduced linear system of nonsingular $M$-matrix type; see \Cref{prop:explicit-column}. Because the constrained limit imposes $\sum_{i=1}^n \bar z_n^{\smallertext S,i}=0$, every non-zero limit vector must have mixed signs. When the $\rho_i$ share a common sign, \Cref{prop:signs-large-gp-column} shows that the only alternatives are the degenerate case $\bar z_n^{\smallertext S}\equiv 0$ and a genuinely mixed-sign configuration.

\begin{figure}[ht!]
  \centering
  \includegraphics[width=.45\linewidth]{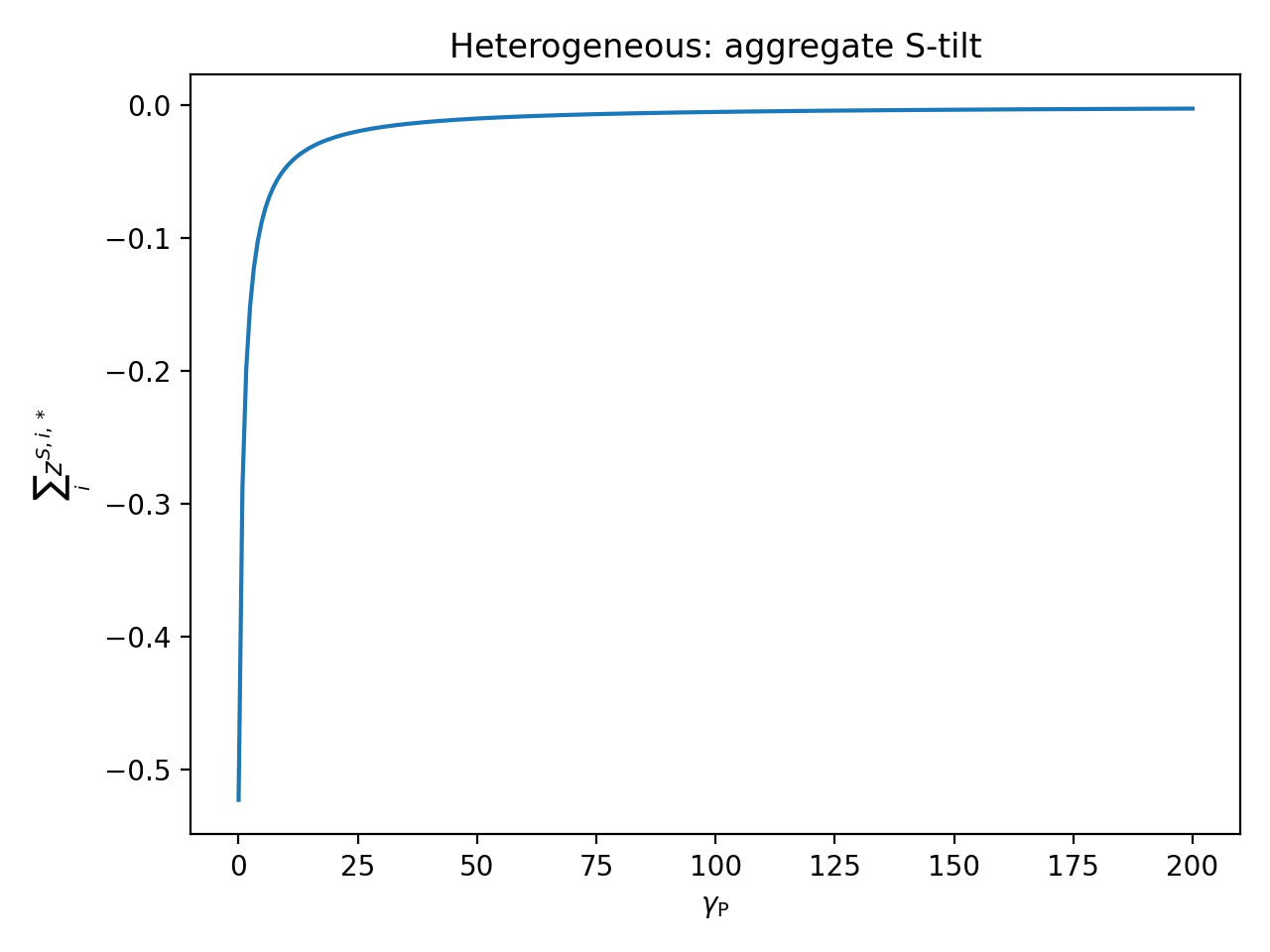}\hfill
  \includegraphics[width=.45\linewidth]{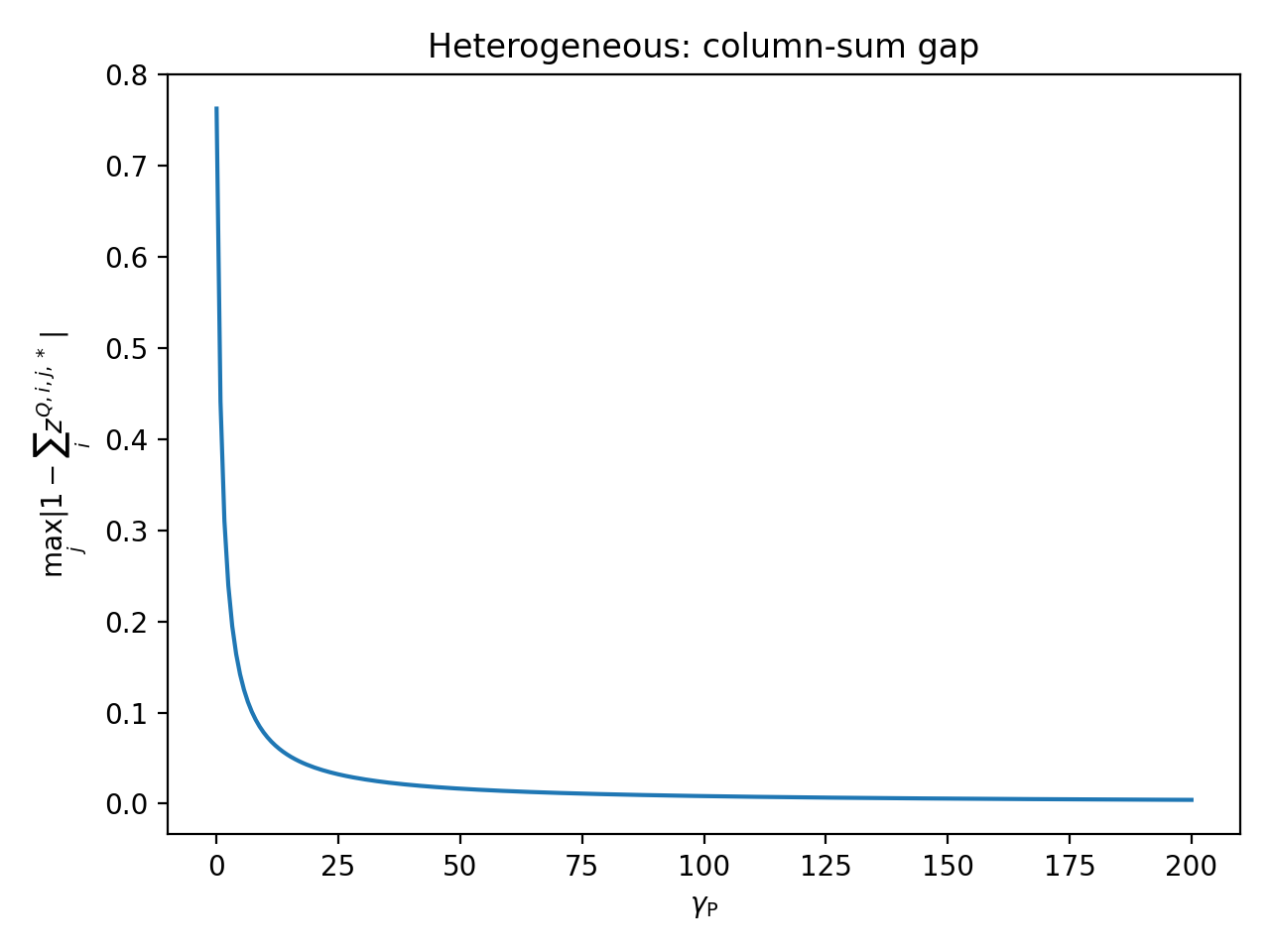}
  \caption{\small Heterogeneous economy. Left: aggregate $S$-tilt $\sum_i z^{\smallertext{S},i,\star}$ tends to $0$ as $\gamma_\smalltext{\rm P}\uparrow\infty$.
  Right: the column identity $\sum_{i=1}^n z^{\smallertext{Q},i,j,\star}=1$ tightens with $\gamma_\smalltext{\rm P}$.}
  \label{fig:hetero-penalty-limit}
\end{figure}

\begin{figure}[ht!]
  \centering
  \includegraphics[width=.6\linewidth]{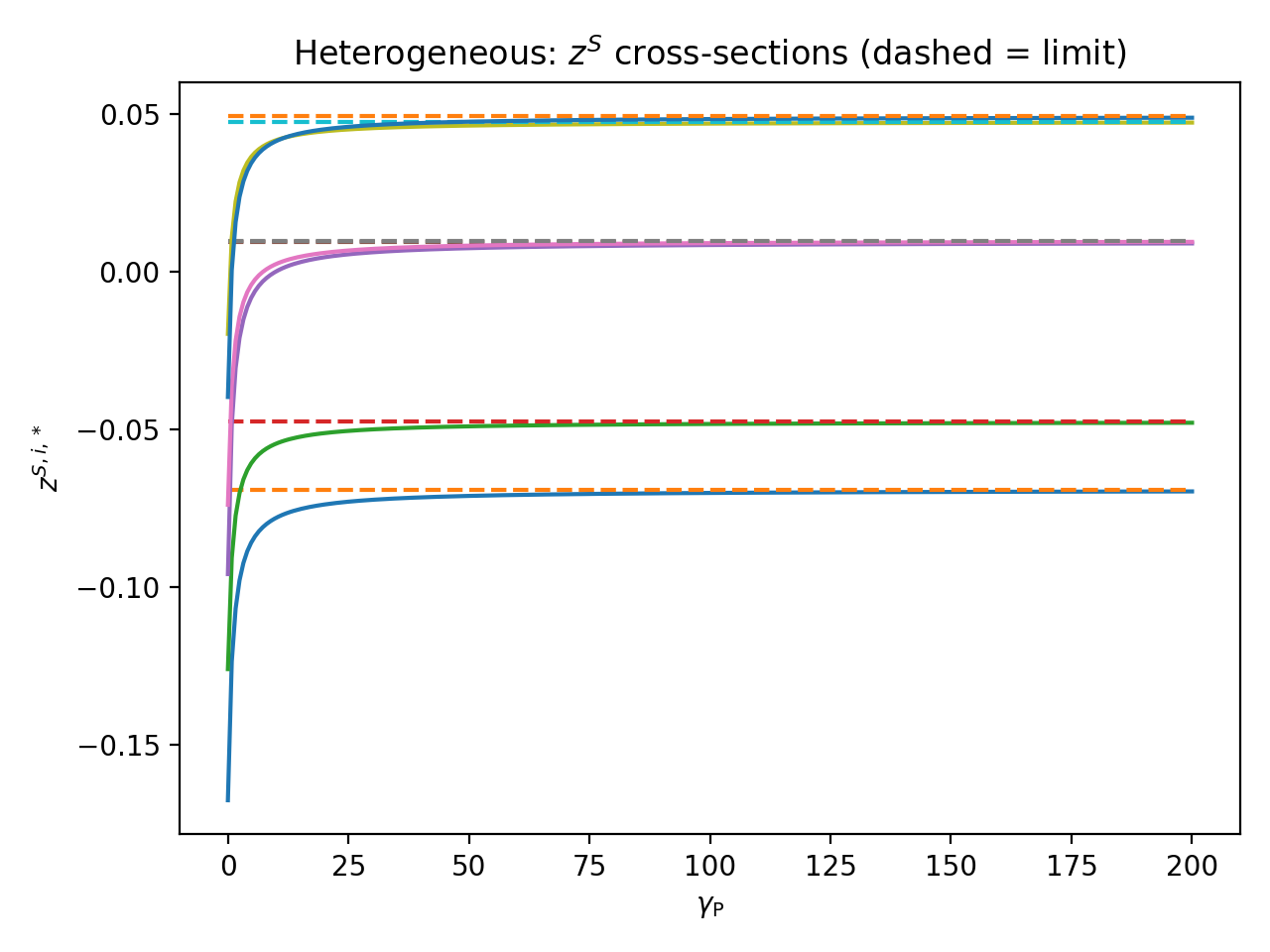}
  \caption{\small Heterogeneous economy. Cross-sections $z^{\smallertext{S},i,\star}(\gamma_\smalltext{\rm P})$ (solid) and constrained limit $\bar z_n^{\smallertext S}$ (dashed), see \Cref{subsec:large-gP-explicit}.}
  \label{fig:hetero-zS-cross}
\end{figure}

\medskip
\emph{Diagonal entries and the possibility of a malus.} The diagonal in the limit splits into a \emph{positive baseline} and a \emph{signed correlation‑balancing correction}
\[
\bar z_n^{\smallertext Q,i,i}
= p_{i,i}\underbrace{\bigg(\frac{1}{\Theta_{i,n}}(1-\zeta_i)+\frac{1}{c_i}\bigg)}_{\text{baseline }C_{i,n}^0>0}
- p_{i,i}\frac{\sigma}{\sqrt{n}}\underbrace{\bigg(\frac{a_i}{\Theta_{i,n}}+\gamma_i\rho_i\nu_i\bigg)}_{\text{has sign }\mathrm{sgn}(\rho_\smalltext{i})}\bar z_n^{\smallertext S,i}.
\]
By \Cref{prop:diag-sign-large-gp}.$(i)$, if the row keeps the \emph{standard} $S$–tilt ($\mathrm{sgn}(\bar z_n^{\smallertext S,i})=-\mathrm{sgn}(\rho_i)$), then $\bar z_n^{\smallertext Q,i,i}>0$; no diagonal flip is possible in that tail.  
Conversely, \Cref{prop:diag-sign-large-gp}.$(ii)$ shows that a negative diagonal at high $\gamma_\smallertext{\rm P}$ \emph{requires} the row's $S$-tilt to align in sign with $\rho_i$ (\emph{i.e.}, a non-standard $S$-tilt for that row). The relevant mechanism is therefore not pairwise correlation with other agents' signals, but an excessively large exposure of row $i$ to the \emph{common traded factor} relative to the rest of the group. In the present unconstrained model, this is not merely a weak bonus reduction: because
\[
a^{\star,i}=\frac{z^{\smallertext Q,i,i,\star}}{c_i},
\]
a negative diagonal means that the induced action-loading itself changes sign. The malus should therefore be interpreted as a \emph{signed disclosure / reporting-distortion prescription}: the principal uses row $i$ to offset residual common-factor exposure rather than to push its own signal upward in the usual direction.
\medskip

This interpretation is economically coherent only because the control set is unconstrained and allows negative actions. If one wishes to rule out such behaviour institutionally, the natural next step is to impose the constraint $\alpha^i\ge0$, in which case the negative-diagonal region should be read as a boundary phenomenon of the constrained problem rather than as an interior optimum. In the present paper, the correct reading is therefore that sufficiently strong principal risk aversion may induce a \emph{sign reversal in the action-loading} of rows whose exposure to the common traded factor is sufficiently large relative to the rest of the group.

\begin{figure}[ht!]
  \centering
  \includegraphics[width=.62\linewidth]{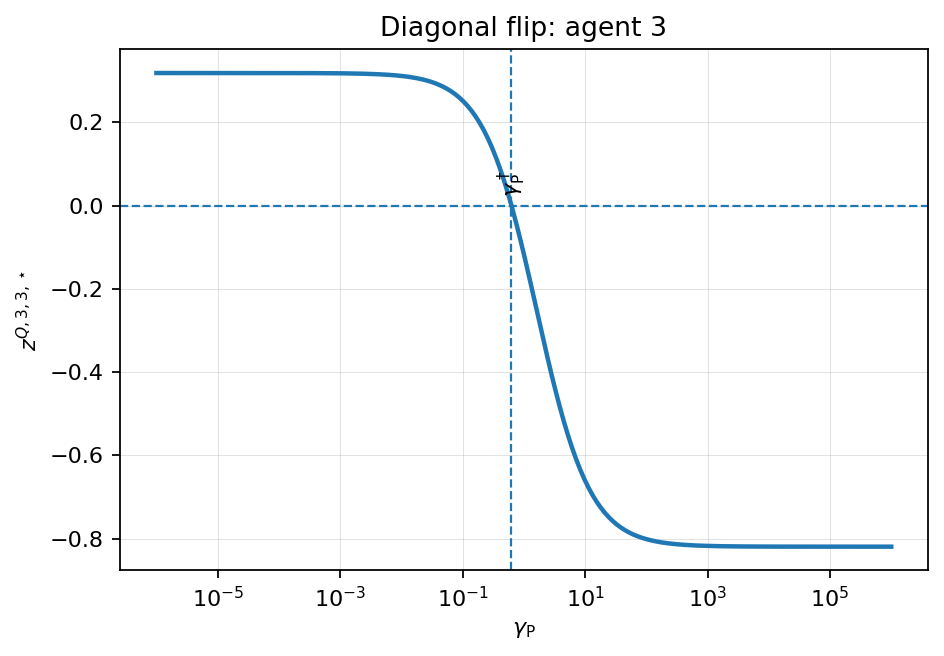}
  \caption{\small Dedicated four-agent calibration from \Cref{tab:diag-flip-params}. The diagonal $z^{\smallertext{Q},3,3,\star}$ crosses zero at $\gamma_\smalltext{\rm P}^\dagger\approx 0.629$ and tends to a negative limit as $\gamma_{\rm P}\uparrow\infty$.}
  \label{fig:diag-flip}
\end{figure}

\begin{appendix}
\section{Technical proofs}

\begin{proof}[Proof of Proposition \ref{prop:FOC-block}]
We keep the $n$-subscripts throughout this proof, so that the dependence on the team size remains visible at every step.

\medskip
\emph{Step $1$: {\rm FOC} in $z^{\smallertext{Q}}$.} Since $f$ is affine-quadratic, its Hessian is constant, and the first-order conditions follow from direct differentiation. We have
	\begin{gather*}
		\partial^2_{z^{\smalltext{Q}\smalltext{,}\smalltext{i}\smalltext{,}\smalltext{j}}z^{\smalltext{Q}\smalltext{,}\smalltext{k}\smalltext{,}\smalltext{\ell}}} f(z^\smallertext{Q},z^\smallertext{S})
		=
		-\frac{1}{nc_i} \mathbf{1}_{\{i=j=k=\ell\}}
		-\frac{  \gamma_i  \nu_j^2   }{n} \mathbf{1}_{\{i=k\}}\mathbf{1}_{\{j=\ell\}} 
		-\frac{\gamma_{\smallertext{\rm P}} } {n^2} { \nu_j^2 \mathbf{1}_{\{\ell = j\}} }, \\
		\partial^2_{z^{\smalltext{S}\smalltext{,}\smalltext{k}}z^{\smalltext{Q}\smalltext{,}\smalltext{i}\smalltext{,}\smalltext{j}}} f(z^\smallertext{Q},z^\smallertext{S})=-\frac{\gamma_i\rho_j\nu_j\sigma}{n^{3/2}}\mathbf{1}_{\{i=k\}}-\frac{\gamma_{\smallertext{\rm P}} {  \rho_j \nu_j} \sigma}{n^{5/2}},\\
		\partial^2_{z^{\smalltext{S}\smalltext{,}\smalltext{i}}z^{\smalltext{S}\smalltext{,}\smalltext{j}}} f(z^\smallertext{Q},z^\smallertext{S})=-\frac{\gamma_i\sigma^2}{n}\mathbf{1}_{\{i=j\}} - \frac{\gamma_{\smallertext{\rm P}}\sigma^2}{n^2}.
	\end{gather*}
	A bit of notation: $\partial_{z^\smalltext{Q}}f$ is an $n\times n$ matrix, $\partial^2_{z^\smalltext{S}z^{\smalltext{S}}}f$ is an $n\times n$ matrix, $\partial^2_{z^\smalltext{S}z^{\smalltext{Q}}}f$ is an $n^2\times n$ matrix, $\partial^2_{z^\smalltext{Q}z^{\smalltext{Q}}}f$ is an $n^2\times n^2$ matrix, and the Hessian matrix of $f$, $D^2f$, is an $(n^2+n)\times (n^2+n)$ matrix with
	\begin{gather*}
	\big(\partial_{z^{\smalltext{Q}}}f\big)^{i,j}\coloneqq \partial_{z^{\smalltext{Q}\smalltext{,}\smalltext{i}\smalltext{,}\smalltext{j}}}f,\; (i,j)\in\{1,\dots,n\}^2,\;
	\big(\partial^2_{z^\smalltext{S}z^{\smalltext{S}}}f\big)^{i,j}\coloneqq \partial^2_{z^{\smalltext{S}\smalltext{,}\smalltext{i}}z^{\smalltext{S}\smalltext{,}\smalltext{j}}} f,\; (i,j)\in\{1,\dots,n\}^2,\\
	\partial^2_{z^\smalltext{S}z^{\smalltext{Q}}}f\coloneqq\begin{pmatrix} \frG_1f\\
	\vdots\\
	\frG_nf
	\end{pmatrix},\;
	\partial^2_{z^\smalltext{Q}z^{\smalltext{Q}}}f\coloneqq\begin{pmatrix}H_{1,1}f & \cdots &H_{1,n}f \\
	\vdots & \ddots & \vdots\\
H_{n,1}f & \cdots &H_{n,n}f
	\end{pmatrix},\\
	(\frG_kf)^{i,j}\coloneqq \partial^2_{z^{\smalltext{Q}\smalltext{,}\smalltext{i}\smalltext{,}\smalltext{k}}z^{\smalltext{S}\smalltext{,}\smalltext{j}}}f,\; (H_{i,j}f)^{k,\ell}\coloneqq \partial^2_{z^{\smalltext{Q}\smalltext{,}\smalltext{k}\smalltext{,}\smalltext{i}}z^{\smalltext{Q}\smalltext{,}\smalltext{\ell}\smalltext{,}\smalltext{j}}}f, \; (i,j,k,\ell)\in\{1,\dots,n\}^4,\\
	D^2f\coloneqq \begin{pmatrix}
	\partial^2_{z^\smalltext{Q}z^{\smalltext{Q}}}f &\partial^2_{z^\smalltext{S}z^{\smalltext{Q}}}f\\
	\big(\partial^2_{z^\smalltext{S}z^{\smalltext{Q}}}f\big)^\smallertext{\top} & \partial^2_{z^\smalltext{S}z^{\smalltext{S}}}f
	\end{pmatrix},
	\end{gather*}
	where derivation of matrices should be understood coordinate-wise. Because the first Hessian block is diagonal in the column index $j$, the compact formula uses $N^2\otimes \Gamma$ (and not $(\nu\nu^\top)\otimes\Gamma$) under the column-wise vectorisation convention. With these notations, the compact matrix formulas announced in the statement are immediate.

\medskip
Next, let $q_{i,j}\coloneqq \nu_j z^{\smallertext{Q},i,j}$ and $q_i=(q_{i,1},\dots,q_{i,n})^\top$, for $(i,j)\in\{1,\dots,n\}^2$.
From the definition of $f$
\[
\frac{\partial f}{\partial z^{\smallertext{Q},i,j}}(z^{\smallertext{Q}},z^{\smallertext{S}})
=-\frac{1}{n}\gamma_i\nu_j^2 z^{\smallertext{Q},i,j}
-\frac{\gamma_i\sigma}{n\sqrt n}\rho_j\nu_j z^{\smallertext{S},i}
+\frac{\gamma_{\smallertext{\rm P}}}{n^2}\nu_j\Bigg({ \nu_j - \nu_j \sum_{\ell=1}^n z^{\smallertext{Q},\ell, j} - \frac{\rho_j\sigma}{\sqrt n}\sum_{m=1}^n z^{\smallertext{S},m} } \Bigg)
-\frac{ z^{\smallertext{Q},i,i} - 1}{n c_i}\mathbf 1_{\{j=i\}}.
\]
{
Setting this to $0$ and writing $K_{i,n}\coloneqq \nu_i-\sum_{\ell=1}^n q_{\ell,i}-\rho_i\frac{\sigma}{\sqrt{n}}\sum_{m=1}^n z^{\smallertext{S},m}$, we obtain
\[
-\frac{\gamma_i}{n}q_{i,j}-\frac{\gamma_i\sigma}{n^{3/2}}\rho_j z^{\smallertext{S},i}+\frac{\gamma_{\smallertext{\rm P}}}{n^2}K_{j,n}-\frac{1}{n c_i}\frac{q_{i,i}}{\nu_i^2} \mathbf 1_{\{j=i\}} + \frac{1}{n c_i\nu_i}\mathbf 1_{\{j=i\}}=0.
\]
Thus, for $j\neq i$
\[
q_{i,j}=\frac{\gamma_{\smallertext{\rm P}}}{n\gamma_i}K_{j,n} - \frac\sigma{\sqrt{n}}\rho_j z^{\smallertext{S},i}=\alpha_{i,n} K_{j,n} - \frac{\sigma}{\sqrt{n}}\rho_j z^{\smallertext{S},i},
\]
and for $j=i$
\[
\bigg(\gamma_i+\frac{1}{c_i\nu_i^2}\bigg)q_{i,i}=\frac{\gamma_{\smallertext{\rm P}}}{n}K_{i,n} - \gamma_i\frac\sigma{\sqrt{n}}\rho_i z^{\smallertext{S},i} + \frac{1}{c_i\nu_i}.
\]
that is to say with our notations 
\[
A_i q_{i,i}=\frac{\gamma_{\smallertext{\rm P}}}{n}K_{i,n} - \gamma_i\frac\sigma{\sqrt{n}}\rho_i z^{\smallertext{S},i} + \frac{1}{c_i\nu_i},
\]
This proves the formulas for the optimal $z^\smallertext{Q}$ once $K_n$ and the optimal $z^{\smallertext{S}}$ are known.
}

\medskip
\emph{Step $2$: solve for $K$ in terms of $z^{\smallertext{S}}$.} Fix $j\in\{1,\dots,n\}$. Summing $q_{i,j}$ over $i\in\{1,\dots,n\}$ gives
\[
\sum_{i=1}^n q_{i,j}
=\Bigg(\sum_{i\in\{1,\dots,n\}\setminus\{j\}}\alpha_{i,n}+\frac{\gamma_{\smallertext{\rm P}}}{nA_j}\Bigg)K_{j,n} - \frac\sigma{\sqrt{n}} \rho_j\sum_{i\in\{1,\dots,n\}\setminus\{j\}}z^{\smallertext{S},i}
-\frac{\gamma_j\sigma\rho_j}{\sqrt{n}A_j}z^{\smallertext{S},j}
+\frac{1}{A_j c_j\nu_j}.
\]
But we also have by definition
\[
\sum_{i=1}^n q_{i,j}=\nu_j-K_{j,n}-\rho_j\frac{\sigma}{\sqrt{n}}\sum_{m=1}^nz^{\smallertext{S},m},
\]
so that rearranging yields
\[
\kappa_{j,n} K_{j,n}=\nu_j-\frac{1}{A_jc_j\nu_j}-\rho_j\frac{\sigma}{\sqrt{n}}\bigg(1-\frac{\gamma_j}{A_j}\bigg)z^{\smallertext{S},j}.
\]
which is exactly the stated affine form $K_{j,n}=d_{j,n}-m_{j,n} z^{\smallertext{S},j}$.

\medskip
\emph{Step $3$: {\rm FOC} in $z^{\smallertext{S}}$.} Differentiating $f$ w.r.t. $z^{\smallertext{S},k}$ gives
\[
0=\frac{\partial f}{\partial z^{\smallertext{S},k}}(z^{\smallertext{Q}},z^{\smallertext{S}})
=-\frac{\gamma_k\sigma^2}{n}z^{\smallertext{S},k}
-\frac{\gamma_k\sigma}{n\sqrt n}\sum_{j=1}^n \rho_j\nu_j z^{\smallertext{Q},k,j}
+\frac{\gamma_{\smallertext{\rm P}}\sigma}{n^2\sqrt n}\sum_{i=1}^n \rho_i K_{i,n}
-\frac{\gamma_{\smallertext{\rm P}}\sigma^2}{n^3}\sum_{i=1}^n (1-\rho_i^2)\sum_{m=1}^n z^{\smallertext{S},m}.
\]
Using $q_{k,j}=\nu_j z^{\smallertext{Q},k,j}$, this is
\begin{equation}\label{eq:lllll}
0=-\gamma_k\frac{\sigma^2}n z^{\smallertext{S},k}-\frac{\gamma_k\sigma}{n^{3/2}}\rho^\top q_k+\frac{\gamma_{\smallertext{\rm P}}\sigma}{n^{5/2}}\sum_{i=1}^n \rho_i K_{i,n}-\frac{\gamma_{\smallertext{\rm P}}\sigma^2}{n^3}\sum_{i=1}^n(1-\rho_i^2)\mathrm{1}_n^\top z^{\smallertext{S}}.
\end{equation}
From Step~$1$
\[
\rho^\top q_k=\sum_{j\in\{1,\dots,n\}\setminus\{ k\}}\rho_j\bigg(\alpha_{k,n} K_{j,n}-\frac\sigma{\sqrt{n}}\rho_j z^{\smallertext{S},k}\bigg)+\rho_kA_k^{\smallertext{-}1}\bigg(\frac{\gamma_{\smallertext{\rm P}}}{n}K_{k,n}-\gamma_k\frac\sigma{\sqrt{n}}\rho_k z^{\smallertext{S},k}+\frac{1}{c_k\nu_k}\bigg).
\]
Collecting coefficients we obtain
\[
\rho^\top q_k=\frac{\gamma_{\smallertext{\rm P}}}{n\gamma_k}\sum_{j=1}^n\rho_jK_{j,n} -\frac{\gamma_{\smallertext{\rm P}}\rho_k}{n\gamma_k}\bigg(1-\frac{\gamma_k}{A_k}\bigg) K_{k,n} - \frac{\sigma}{\sqrt{n}}\Bigg(\sum_{j=1}^n\rho_j^2-\bigg(1-\frac{\gamma_k}{A_k}\bigg)\rho^2_k\Bigg)z^{\smallertext{S},k} + \frac{\rho_k}{A_k c_k\nu_k}.
\]
Using this in \Cref{eq:lllll}, we deduce
\begin{equation}\label{eq:lllll2}
0=-\gamma_k\frac{\sigma^2}n z^{\smallertext{S},k}+\frac{\gamma_{\smallertext{\rm P}}\sigma\rho_k}{n^{5/2}}\bigg(1-\frac{\gamma_{k}}{A_k}\bigg) K_{k,n}+\frac{\gamma_k\sigma^2}{n^{2}}\Bigg(\sum_{j=1}^n\rho_j^2-\bigg(1-\frac{\gamma_k}{A_k}\bigg)\rho_k^2\Bigg)z^{\smallertext{S},k} -\frac{\gamma_k\sigma}{n^{3/2}} \frac{\rho_k}{A_k c_k\nu_k}-\frac{\gamma_{\smallertext{\rm P}}\sigma^2}{n^3}\sum_{i=1}^n(1-\rho_i^2)\mathrm{1}_n^\top z^{\smallertext{S}}.
\end{equation}
We then once more use Step $1$ to write $K_{k,n}=d_{k,n}-m_{k,n}z^{\smallertext{S},k}$, and finally deduce
\[
\mu_{k,n}z^{\smallertext{S},k}+\lambda_n{1}_n^\top z^{\smallertext{S}}=\ell_{k,n}, \; \text{\rm or in matrix form}\; \big(D_n+\lambda_n {1}_n{1}_n^\top\big)z^{\smallertext{S}}=\ell_n,
\]
which is again the desired equation, that we will solve explicitly in the next step.

\medskip
\emph{Step $4$: closed form solution via {\rm Sherman–Morrison–Woodbury}.} We have an immediate rank one factorisation, hence since $D_n$ is diagonal and positive definite (the sign is obvious since for any $i\in\{1,\dots,n\}$, $A_i>\gamma_i$ and $\rho_i^2\leq 1$), the Sherman–Morrison–Woodbury formula gives
\[
\big(D_n+\lambda_n {1}_n{1}_n^\top\big)^{-1}=S_n-y_ns_ns_n^\top,
\]
which gives the desired formula. Then, it is immediate to deduce the optimal $z^{\smallertext{Q},\star}$.

\medskip
\emph{Step $5$: verification.} Finally, substitute the candidate $(z^{\smallertext{Q},\star}, z^{\smallertext{S},\star})$ into \eqref{eq:FOC}. A direct check shows that it solves the first-order conditions, and therefore coincides with the unique maximiser.
\end{proof}

\begin{proof}[Proof of Proposition \ref{prop:homogeneous-closed-form}]
\emph{Step 1: row-wise {\rm FOC} specialised to the homogeneous case.}
From \Cref{prop:FOC-block} (with $c_i=c$, $\gamma_i=\gamma$, $\nu_i=\nu$, $\rho_i=\rho$), for any $(i,j)\in\{1,\dots,n\}^2$ one has
\[
q_{i,j,n}^\star\coloneqq \nu z_n^{\smallertext Q,i,j,\star}
=
\begin{cases}
\alpha_n K_{j,n}(z_n^{\smallertext S,\star})-\beta_n z_{s,n}^\star,
\; j\neq i,\\[0.4em]
A^{-1}\bigg(\dfrac{\gamma_\smallertext{\rm P}}{n}K_{i,n}(z_n^{\smallertext S,\star})-\gamma\beta_n z_{s,n}^\star+\dfrac{1}{c\nu}\bigg),
\; j=i.
\end{cases}
\]
By symmetry, all column residuals coincide; write their common value as $K_n^\star$. Then there are two numbers $q_{o,n}$ and $q_{d,n}$ such that
\begin{gather*}
q_{i,j,n}^\star=q_{o,n}\coloneqq \alpha_n K_n^\star-\beta_n z_{s,n}^\star,
\; i\neq j,\\
q_{i,i,n}^\star=q_{d,n}\coloneqq A^{-1}\bigg(\frac{\gamma_\smallertext{\rm P}}{n}K_n^\star-\gamma\beta_n z_{s,n}^\star+\frac{1}{c\nu}\bigg).
\end{gather*}

\emph{Step 2: express $K_n^\star$ in terms of $z_{s,n}^\star$.}
By definition,
\[
K_n^\star
=\nu-\sum_{\ell=1}^n q_{\ell,i,n}^\star-\rho\frac{\sigma}{\sqrt n}\sum_{m=1}^n z_n^{\smallertext S,m,\star}
=\nu-\big((n-1)q_{o,n}+q_{d,n}\big)-n\beta_n z_{s,n}^\star.
\]
Substituting $q_{o,n}$ and $q_{d,n}$ gives
\[
\bigg(1+(n-1)\alpha_n+\frac{\gamma_\smallertext{\rm P}}{nA}\bigg)K_n^\star
=\frac{\gamma\nu}{A}+\beta_n\bigg(\frac{\gamma}{A}-1\bigg)z_{s,n}^\star,
\]
that is,
\[
K_n^\star=\frac{1}{\kappa_n}\bigg(\frac{\gamma\nu}{A}-\beta_n\delta z_{s,n}^\star\bigg),
\]
which is \eqref{eq:K-star}.

\medskip
\emph{Step 3: the $S$-block {\rm FOC} and the formula for $z_{s,n}^\star$.}
The $z^{\smallertext S}$-{\rm FOC} in \Cref{prop:FOC-block} reduces here to
\[
0
=-\frac{\gamma\sigma^2}{n}z_{s,n}^\star
-\frac{\gamma\sigma}{n\sqrt n}\rho\nu\sum_{j=1}^n z_n^{\smallertext Q,i,j,\star}
+\frac{\gamma_\smallertext{\rm P}\sigma}{n\sqrt n}\rho K_n^\star
-\frac{\gamma_\smallertext{\rm P}\sigma^2}{n}(1-\rho^2)z_{s,n}^\star.
\]
Using
\[
\sum_{j=1}^n z_n^{\smallertext Q,i,j,\star}=\frac{(n-1)q_{o,n}+q_{d,n}}{\nu},
\]
we obtain
\begin{equation}
\label{eq:eq-zs-homogeneous-indexed}
\sigma\big(\gamma+\gamma_\smallertext{\rm P}(1-\rho^2)\big)z_{s,n}^\star
+\frac{\gamma\rho}{\sqrt n}\big((n-1)q_{o,n}+q_{d,n}\big)
=\frac{\gamma_\smallertext{\rm P}\rho}{\sqrt n}K_n^\star.
\end{equation}
Now
\[
(n-1)q_{o,n}+q_{d,n}
=(\kappa_n-1)K_n^\star-\beta_n\bigg((n-1)+\frac{\gamma}{A}\bigg)z_{s,n}^\star+\nu\delta.
\]
Insert this identity and \eqref{eq:K-star} into \eqref{eq:eq-zs-homogeneous-indexed}. After collecting terms one gets
\[
\Delta_n(\gamma_\smallertext{\rm P})z_{s,n}^\star
=-\frac{\rho\gamma}{A c \sigma \nu \sqrt n}\bigg(1-\frac{\gamma_\smallertext{\rm P}}{n\widetilde\kappa_n}\bigg),
\]
where the denominator simplifies as
\[
\gamma+\gamma_\smallertext{\rm P}(1-\rho^2)
-\frac{\gamma\rho^2}{n}\bigg((n-1)+\frac{\gamma}{A}\bigg)
+\frac{\gamma_\smallertext{\rm P}\rho^2}{n^2}\frac{(A-\gamma)^2}{A\widetilde\kappa_n}
=(\gamma+\gamma_\smallertext{\rm P})(1-\rho^2)
+\frac{\gamma\rho^2\delta}{n}
+\frac{\gamma_\smallertext{\rm P}\rho^2\delta^2}{n^2\kappa_n}.
\]
This proves \eqref{eq:zs-star}. Finally, \eqref{eq:zozd-star} follows by inserting $z_{s,n}^\star$ and $K_n^\star$ into the row-wise formulas.
\end{proof}

\begin{lemma}[Full row rank of $O_n$]\label{lem:row-rank}
The $(n+1)\times(n^2+n)$ matrix of $O_n$ has full row rank $n+1$.
\end{lemma}

\begin{proof}
Let $u=(u_1,\dots,u_n,u_{n+1})^\top\in\R^{n+1}$ satisfy $u^\top O_n=0$. For $(i,j)\in\{1,\dots,n\}^2$, denote by $\mathbf e_{(i,j)}\in\R^{n^2+n}$ the canonical basis vector corresponding to the variable $z^{\smallertext Q,i,j}$, and for $k\in\{1,\dots,n\}$ denote by $\mathbf e_{S,k}\in\R^{n^2+n}$ the one corresponding to the variable $z^{\smallertext S,k}$.

\medskip
Fix $(i,j)\in\{1,\dots,n\}^2$. The column of $O_n$ associated with $z^{\smallertext Q,i,j}$ has a single non-zero entry, equal to $\nu_j$ in row $j$. Hence
\[
u^\top O_n\mathbf e_{(i,j)}=u_j\nu_j=0.
\]
Since $\nu_j>0$, it follows that $u_j=0$ for every $j\in\{1,\dots,n\}$. Now fix $k\in\{1,\dots,n\}$. The column of $O_n$ associated with $z^{\smallertext S,k}$ has first $n$ entries
\[
\bigg(\frac{\rho_1\sigma}{\sqrt n},\dots,\frac{\rho_n\sigma}{\sqrt n}\bigg)^\top,
\]
and last entry $1$. Therefore
\[
u^\top O_n\mathbf e_{S,k}
=
\frac{\sigma}{\sqrt n}\sum_{j=1}^n \rho_j u_j + u_{n+1}
=
u_{n+1}
=
0.
\]
Thus $u=0$, and the rows of $O_n$ are linearly independent. Hence $O_n$ has full row rank $n+1$.
\end{proof}

\begin{proof}[Proof of Proposition \ref{thm:penalty-limit}]
Evaluating $f_\gamma$ at $x_n^\star$ (for which $O_nx_n^\star=o_n$) gives
\[
f_\gamma(x_{\gamma,n}) \ge f_\gamma(x_n^\star)=g(x_n^\star).
\]
Hence
\[
\frac{\gamma_{\smallertext{\rm P}}}{2}\|r_{\gamma,n}\|_{\smallertext{W}_\smalltext{n}}^2
= g(x_{\gamma,n})-f_\gamma(x_{\gamma,n})
\le \sup_{x\in\R^{\smalltext{n}^\tinytext{2}\smalltext{+}\smalltext{n}}} g(x)-g(x_n^\star)
\eqqcolon C_{0,n}<+\infty,
\]
because $g$ is a strictly concave quadratic and therefore bounded above. Thus
\begin{equation}\label{eq:residual-O1overgamma-fixed}
\|r_{\gamma,n}\|_{\smallertext{W}_\smalltext{n}}^2
\le \frac{2C_{0,n}}{\gamma_{\smallertext{\rm P}}},
\end{equation}
and this already implies
\[
\Phi_n(x_{\gamma,n})=\|r_{\gamma,n}\|_{\smallertext{W}_\smalltext{n}}^2=\Oc(1/\gamma_{\smallertext{\rm P}}).
\]

Subtract \eqref{eq:KKT-constrained} from \eqref{eq:stat-penalty} to get
\[
H_n(x_{\gamma,n}-x_n^\star)
=
O_n^\top\big(\iota_n^\star+\gamma_{\smallertext{\rm P}}W_n r_{\gamma,n}\big).
\]
Since $O_nx_n^\star=o_n$, we have
\[
r_{\gamma,n}=O_n(x_{\gamma,n}-x_n^\star).
\]
Multiplying the stationarity difference by $O_nH_n^{-1}$ yields
\[
r_{\gamma,n}
=
O_nH_n^{-1}O_n^\top\big(\iota_n^\star+\gamma_{\smallertext{\rm P}}W_n r_{\gamma,n}\big).
\]
By \Cref{lem:row-rank}, $O_n$ has full row rank. Since $H_n\prec0$, the matrix $O_nH_n^{-1}O_n^\top$ is symmetric negative definite. Define
\[
\Sc_n\coloneqq -W_n^{1/2}O_nH_n^{-1}O_n^\top W_n^{1/2}\succ0,
\;
 y_{\gamma,n}\coloneqq W_n^{1/2}r_{\gamma,n}.
\]
Then
\begin{equation}\label{eq:system-r-fixed}
\big(\mathrm{I}_{n+1}+\gamma_{\smallertext{\rm P}}\Sc_n\big)y_{\gamma,n}
=
W_n^{1/2}O_nH_n^{-1}O_n^\top\iota_n^\star.
\end{equation}
Hence
\[
\|y_{\gamma,n}\|
\le
\frac{1}{1+\gamma_{\smallertext{\rm P}}\lambda_{\min}(\Sc_n)}
\big\|W_n^{1/2}O_nH_n^{-1}O_n^\top\iota_n^\star\big\|
=
\Oc(1/\gamma_{\smallertext{\rm P}}).
\]
Since $\|r_{\gamma,n}\|_{\smallertext{W}_\smalltext{n}}=\|y_{\gamma,n}\|$, we obtain the sharper estimate
\begin{equation}\label{eq:residual-sharp-fixed}
\|O_nx_{\gamma,n}-o_n\|_{\smallertext{W}_\smalltext{n}}=\Oc(1/\gamma_{\smallertext{\rm P}}).
\end{equation}
Because $W_n\succ0$ and $n$ is fixed throughout the limit, the weighted norm $\|\cdot\|_{\smallertext{W}_\smalltext{n}}$ is equivalent to the Euclidean norm on $\R^{n+1}$. Therefore
\[
\|r_{\gamma,n}\|=\Oc(1/\gamma_{\smallertext{\rm P}}).
\]

By definition
\[
\widehat\iota_{\gamma,n}=-\gamma_{\smallertext{\rm P}}W_n r_{\gamma,n},
\]
so \eqref{eq:stat-penalty} rewrites as
\[
\partial g(x_{\gamma,n})+O_n^\top\widehat\iota_{\gamma,n}=0.
\]
Subtracting the KKT relation $\partial g(x_n^\star)+O_n^\top\iota_n^\star=0$ gives
\[
H_n(x_{\gamma,n}-x_n^\star)
=
O_n^\top(\iota_n^\star-\widehat\iota_{\gamma,n}).
\]
Multiplying by $O_nH_n^{-1}$ yields
\[
r_{\gamma,n}
=
O_nH_n^{-1}O_n^\top(\iota_n^\star-\widehat\iota_{\gamma,n}).
\]
Since $O_nH_n^{-1}O_n^\top$ is invertible,
\begin{equation}\label{eq:multiplier-rate-fixed}
\|\widehat\iota_{\gamma,n}-\iota_n^\star\|
\le
\big\|(O_nH_n^{-1}O_n^\top)^{-1}\big\|\|r_{\gamma,n}\|
=
\Oc(1/\gamma_{\smallertext{\rm P}}).
\end{equation}

Let $e_{\gamma,n}\coloneqq x_{\gamma,n}-x_n^\star$. Decompose
\[
e_{\gamma,n}=e_{\gamma,n}^{\parallel}+e_{\gamma,n}^{\perp},
\]
with
\[
O_ne_{\gamma,n}^{\parallel}=r_{\gamma,n},
\;
O_ne_{\gamma,n}^{\perp}=0,
\;
e_{\gamma,n}^{\parallel}\in\mathrm{range}(O_n^\top),
\;
e_{\gamma,n}^{\perp}\in\ker(O_n).
\]
Choose, for instance
\[
e_{\gamma,n}^{\parallel}=O_n^\top(O_nO_n^\top)^{-1}r_{\gamma,n},
\]
which is well defined by \Cref{lem:row-rank}. Then
\[
\|e_{\gamma,n}^{\parallel}\|
\le
\|O_n^\top(O_nO_n^\top)^{-1}\|\|r_{\gamma,n}\|
=
\Oc(1/\gamma_{\smallertext{\rm P}}).
\]
Let $P_{\ker}$ denote the orthogonal projector onto $\ker(O_n)$. Projecting the stationarity difference
\[
H_ne_{\gamma,n}=O_n^\top(\iota_n^\star-\widehat\iota_{\gamma,n})
\]
onto $\ker(O_n)$ yields $P_{\ker}H_ne_{\gamma,n}=0,$ hence
\[
P_{\ker}H_ne_{\gamma,n}^{\perp}=-P_{\ker}H_ne_{\gamma,n}^{\parallel}.
\]
Since the restriction of $H_n$ to $\ker(O_n)$ is negative definite, it is invertible there. Therefore
\[
\|e_{\gamma,n}^{\perp}\|
\le
\big\|\big(H_n|_{\ker(O_n)}\big)^{-1}\big\|\|H_n\|\|e_{\gamma,n}^{\parallel}\|
=
\Oc(1/\gamma_{\smallertext{\rm P}}).
\]
Together with the bound on $e_{\gamma,n}^{\parallel}$, this proves
\[
\|x_{\gamma,n}-x_n^\star\|=\Oc(1/\gamma_{\smallertext{\rm P}}).
\]
Finally, since $g$ is quadratic with Hessian $H_n$, a Taylor expansion around $x_n^\star$ gives
\[
g(x_{\gamma,n})-g(x_n^\star)
=
\partial g(x_n^\star)\cdot e_{\gamma,n}
+\frac12 e_{\gamma,n}^\top H_n e_{\gamma,n}.
\]
Using $\partial g(x_n^\star)=-O_n^\top\iota_n^\star$ from \eqref{eq:KKT-constrained} and $\|e_{\gamma,n}\|=\Oc(1/\gamma_{\smallertext{\rm P}})$, we obtain
\[
\big|g(x_{\gamma,n})-g(x_n^\star)\big|
\le
\|O_n^\top\iota_n^\star\|\|e_{\gamma,n}\|
+\frac12\|H_n\|\|e_{\gamma,n}\|^2
=
\Oc(1/\gamma_{\smallertext{\rm P}}),
\]
which proves $(iv)$.
\end{proof}

\begin{proof}[Proof of Proposition \ref{prop:explicit-column}]
\emph{Step 1: FOC in $z^{\smallertext Q}$.}
Let $\mu_n=(\mu_{1,n},\dots,\mu_{n,n})^\top$ be the Lagrange multipliers for the column constraints and let $\vartheta_n$ be the multiplier for $\sum_{k=1}^n z^{\smallertext S,k}=0$. For fixed row $i\in\{1,\dots,n\}$, the first-order condition yields, coordinatewise in $j\in\{1,\dots,n\}$,
\[
\gamma_i\nu_j^2 z^{\smallertext Q,i,j}+\frac{1}{c_i}\mathbf{1}_{\{j=i\}} z^{\smallertext Q,i,i}
+\frac{\gamma_i\sigma}{\sqrt n}z^{\smallertext S,i}\rho_j\nu_j
= \frac{1}{c_i}\mathbf{1}_{\{j=i\}}+n\mu_{j,n}.
\]
In vector form,
\[
\mathcal H_{i,n} q_i+\frac{\gamma_i\sigma}{\sqrt n}z^{\smallertext S,i}(\rho\odot\nu)
=\frac{1}{c_i}e_i+n\mu_n,
\]
where $q_i\coloneqq(z^{\smallertext Q,i,1},\dots,z^{\smallertext Q,i,n})^\top$. Thus
\begin{equation}
\label{eq:row-solution}
q_i=\mathcal P_{i,n}\bigg(\frac{1}{c_i}e_i-\frac{\gamma_i\sigma}{\sqrt n}z^{\smallertext S,i}(\rho\odot\nu)+n\mu_n\bigg).
\end{equation}

\medskip
\emph{Step 2: enforce the column constraints.}
For a fixed column $j\in\{1,\dots,n\}$, summing \eqref{eq:row-solution} over $i$ gives
\[
\sum_{i=1}^n z^{\smallertext Q,i,j}
=\sum_{i=1}^n p_{i,j}\bigg(n\mu_{j,n}+\frac{1}{c_i}\mathbf{1}_{\{i=j\}}-\frac{\gamma_i\sigma}{\sqrt n}z^{\smallertext S,i}\rho_j\nu_j\bigg)=1.
\]
Using $\sum_{i\neq j}z^{\smallertext S,i}=-z^{\smallertext S,j}$, one obtains
\[
n\mu_{j,n}\Theta_{j,n}+\zeta_j+\frac{\sigma}{\sqrt n}a_j z^{\smallertext S,j}=1,
\]
which proves \eqref{eq:mu-column}.

\medskip
\emph{Step 3: FOC in $z^{\smallertext S}$.}
For each $k\in\{1,\dots,n\}$,
\[
\gamma_k\sigma^2 z^{\smallertext S,k}+\frac{\gamma_k\sigma}{\sqrt n}\sum_{j=1}^n \rho_j\nu_j z^{\smallertext Q,k,j}=n\vartheta_n.
\]
Insert \eqref{eq:row-solution} with $i=k$
\[
\sum_{j=1}^n \rho_j\nu_j z^{\smallertext Q,k,j}
=\frac{\rho_k\nu_k}{c_k}p_{k,k}
-\frac{\gamma_k\sigma}{\sqrt n}z^{\smallertext S,k}\upsilon_{k,n}
+n\sum_{j=1}^n M_{k,j}\mu_{j,n}.
\]
Hence
\[
\gamma_k\sigma^2\varphi_{k,n}z^{\smallertext S,k}
+\gamma_k\sigma\sqrt n\sum_{j=1}^n M_{k,j}\mu_{j,n}
+\frac{\gamma_k\sigma}{\sqrt n}\frac{\rho_k\nu_k}{c_k}p_{k,k}
=n\vartheta_n.
\]
Now substitute \eqref{eq:mu-column}:
\[
\sqrt n\sum_{j=1}^n M_{k,j}\mu_{j,n}
=\frac{1}{\sqrt n}\sum_{j=1}^n\frac{1-\zeta_j}{\Theta_{j,n}}M_{k,j}
-\frac{\sigma}{n}\sum_{j=1}^n\frac{a_jM_{k,j}}{\Theta_{j,n}}z^{\smallertext S,j}.
\]
Therefore
\[
z^{\smallertext S,k}
-\sum_{j=1}^n \frac{1}{\varphi_{k,n}}w_{j,n}a_jM_{k,j}z^{\smallertext S,j}
=\frac{n}{\gamma_k\sigma^2\varphi_{k,n}}\vartheta_n
-\frac{1}{\varphi_{k,n}}\Bigg(
\frac{1}{\sigma\sqrt n}\sum_{j=1}^n\frac{1-\zeta_j}{\Theta_{j,n}}M_{k,j}
+\frac{1}{\sigma\sqrt n}\frac{\rho_k\nu_k}{c_k}p_{k,k}
\Bigg).
\]
This is exactly $(\mathrm{I}_n-L_n)z^{\smallertext S}=V_n\vartheta_n+U_n.$

\medskip
\emph{Step 4: resolvent positivity via a weighted Perron--Frobenius argument.}
Since $a_j=\rho_j\zeta_j/\nu_j$, we have
\[
(L_n)_{k,j}=\frac{\rho_j^2\zeta_j p_{k,j}}{n\varphi_{k,n}\Theta_{j,n}}\ge 0.
\]
Moreover, for each $j$,
\[
(L_n^\top\varphi_n)_j
=\sum_{k=1}^n \varphi_{k,n}(L_n)_{k,j}
=\frac{\rho_j^2\zeta_j}{n\Theta_{j,n}}\sum_{k=1}^n p_{k,j}
=\frac{\rho_j^2\zeta_j}{n}.
\]
On the other hand,
\[
\varphi_{j,n}
=1-\frac{\gamma_j}{n}\upsilon_{j,n}
=1-\frac1n\sum_{m=1}^n\rho_m^2+\frac{\rho_j^2\zeta_j}{n},\; \text{\rm because}\; \frac{\gamma_j\nu_j^2}{\gamma_j\nu_j^2+1/c_j}=1-\zeta_j.
\]
Hence
\[
\varphi_{j,n}-(L_n^\top\varphi_n)_j
=1-\frac1n\sum_{m=1}^n\rho_m^2>0,
\]
since $|\rho_m|<1$ for every $m$. Therefore $L_n^\top\varphi_n<\varphi_n$. Since $L_n\ge0$ and $\varphi_n>0$, the Collatz--Wielandt characterisation of the Perron root for nonnegative matrices (see, \emph{e.g.}, \citeauthor*{horn2013matrix} \cite[Chapter~8]{horn2013matrix}) yields
\[
\spr(L_n)=\spr(L_n^\top)<1.
\]
Consequently
\[
(\mathrm{I}_n-L_n)^{-1}=\sum_{m=0}^\infty L_n^m\ge 0,
\]
so $(\mathrm{I}_n-L_n)$ is a nonsingular $M$-matrix.

\medskip
\emph{Step 5: determine $\vartheta_n^\star$ and conclude.}
The affine constraint $\mathrm{1}_n^\top z^{\smallertext S}=0$ gives
\[
0=\mathrm{1}_n^\top z^{\smallertext S}
=\mathrm{1}_n^\top(\mathrm{I}_n-L_n)^{-1}(V_n\vartheta_n+U_n),
\]
so
\[
\vartheta_n^\star
=-\frac{\mathrm{1}_n^\top(\mathrm{I}_n-L_n)^{-1}U_n}{\mathrm{1}_n^\top(\mathrm{I}_n-L_n)^{-1}V_n}
=-\frac{\mathrm{1}_n^\top v_n}{\mathrm{1}_n^\top u_n}.
\]
This yields $\bar z_n^{\smallertext S}=u_n\vartheta_n^\star+v_n$, and \eqref{eq:zQ-column} follows from \eqref{eq:row-solution} with $\mu_n=\bar\mu_n$ and $z^{\smallertext S}=\bar z_n^{\smallertext S}$. Uniqueness holds by strict concavity of $g$ on the affine space $\frF_n$.
\end{proof}

\section{Numerical parameters}\label{sec:numerical-params}

We report the parameter sets used to generate the figures. Homogeneous runs use the calibration in \Cref{tab:params-homo}. The three heterogeneous penalty-limit figures---\Cref{fig:hetero-penalty-limit,fig:hetero-zS-cross}---use the six-agent calibration in \Cref{tab:params-heteroA}. The diagonal-flip figure uses a separate four-agent calibration, reported in \Cref{tab:diag-flip-params}; it is therefore \emph{not} generated from the heterogeneous baseline in \Cref{tab:params-heteroA}.

\begin{table}[ht!]
\centering
\caption{Homogeneous calibration used for \Cref{fig:homo-all}.\label{tab:params-homo}}
\begin{tabular}{lcc}
\toprule
\small Parameter &\small  Symbol &\small  Value \\
\midrule
\small Number of agents &\small  $n$ & \small $6$ \\
\small Agent risk aversion &\small  $\gamma$ &\small  $1.0$ \\
\small Effort cost scale & \small $c$ & \small $1.2$ \\
\small Signal scale & \small $\nu$ &\small  $1.0$ \\
\small Correlation with $S$ & \small $\rho$ &\small  $0.6$ \\
\small Asset volatility &\small  $\sigma$ &\small  $1.0$ \\
\small Principal risk aversion & \small $\gamma_\smallertext{\rm P}$ & \small varied on $[0,10]$ \\
\bottomrule
\end{tabular}

\vspace{0.5em}
\footnotesize
Grid used in the figures: $\gamma_\smallertext{\rm P}\in\{0,0.25,0.5,\dots,10\}$.
\end{table}

\begin{table}[ht!]
\centering
\caption{\footnotesize Six-agent heterogeneous calibration used for \Cref{fig:hetero-penalty-limit,fig:hetero-zS-cross}.\label{tab:params-heteroA}}
\begin{tabular}{lcccccc}
\toprule
\small $i$ &\small  $1$ &\small  $2$ &\small  $3$ &\small  $4$ &\small  $5$ &\small  $6$ \\
\midrule
\small $c_i$      & \small $1.2$ &\small  $1.0$ &\small  $1.5$ &\small  $1.3$ &\small  $2.0$ &\small  $2.5$ \\
\small $\gamma_i$ &\small  $1.0$ &\small  $1.2$ &\small  $0.9$ &\small  $1.1$ &\small  $1.3$ &\small  $1.0$ \\
\small $\nu_i$    &\small  $1.0$ &\small  $0.9$ &\small  $1.1$ &\small  $1.0$ & \small $1.2$ &\small  $0.8$ \\
\small $\rho_i$   & \small $0.75$ & \small $0.55$ & \small $0.45$ & \small $0.35$ & \small $0.15$ &\small  $0.25$ \\
\midrule
\small $\sigma$   &   \multicolumn{6}{c}{\small $1.0$} \\
\small $n$        &  \multicolumn{6}{c}{\small $6$} \\
\small $\gamma_\smallertext{\rm P}$ &  \multicolumn{6}{c}{\small varied on $[0,40]$} \\
\bottomrule
\end{tabular}

\vspace{0.5em}
\footnotesize
Grid used in the figures: $\gamma_{\rm P}\in\{0,1,2,\dots,40\}$.
\end{table}

\begin{table}[ht!]
\centering
\caption{\footnotesize Dedicated four-agent calibration used for the diagonal-flip experiment.\label{tab:diag-flip-params}}
\begin{tabular}{lcccc}
\toprule
 & \small Agent 1 & \small Agent 2 & \small Agent 3 & \small Agent 4 \\
\midrule
\small $c_i$      &\small  $0.001679$ & \small $0.021845$ & \small $25.315806$ & \small $0.009048$ \\
\small $\gamma_i$ & \small $2.328781$ & \small $1.944496$ & \small $0.538463$ &\small  $1.723001$ \\
\small $\nu_i$    &\small  $1.404647$ & \small $1.945208$ & \small $0.480164$ &\small  $1.553918$ \\
\small $\rho_i$   & \small $-0.746539$ & \small $-0.894259$ & \small $-0.850954$ & \small $-0.605495$ \\
\midrule
\small $\sigma$   & \multicolumn{4}{c}{\small $1.0$} \\
\small $n$        & \multicolumn{4}{c}{\small $4$} \\
\bottomrule
\end{tabular}
\end{table}

\end{appendix}

{\small
\bibliography{bibliographyDylan}}

\end{document}